\documentclass[sort&compress,preprint,12pt]{elsarticle}

\usepackage{amsthm,amssymb,amsmath,amsfonts,amsthm,amscd,graphicx} 
\usepackage[abbrev]{amsrefs}
\usepackage{tikz,mathtools,relsize}
\usepackage{mathtools}
\usepackage{upgreek}
\usepackage{comment}
\usepackage{lipsum}
\usepackage[normalem]{ulem}
\usepackage{mwe}
\usepackage[font=small]{caption} 
\usepackage[font=small]{subcaption}
\usepackage[export]{adjustbox}
\usepackage{multicol,color}
\usepackage{lineno}
\usepackage{xcolor}
\usepackage{tikz}

\usetikzlibrary{arrows}

\tikzset{
    vertex/.style={circle,draw,minimum size=2.5em},
    edge/.style={->,> = latex'}
}

\usepackage{hyperref}
\graphicspath{ {./Latex_source/} }

\allowdisplaybreaks

\newtheorem{lemma}{Lemma}
\newtheorem{theorem}{Theorem}
\newtheorem{corollary}{Corollary}[theorem]
\newtheorem{remark}{Remark}
\newtheorem{assumption}{Assumption}
\newtheorem{definition}{Definition}[section]

\newcommand{\be}{\begin{equation}}
\newcommand{\ee}{\end{equation}}
\newcommand{\bea}{\begin{eqnarray}}
\newcommand{\eea}{\end{eqnarray}}
\newcommand{\beas}{\begin{eqnarray*}}
\newcommand{\eeas}{\end{eqnarray*}}
\newcommand{\ba}{\begin{array}}
\newcommand{\ea}{\end{array}}

\def\al#1\eal{\begin{align}#1\end{align}}
\def\als#1\eals{\begin{align*}#1\end{align*}}

\newcommand{\norm}[1]{\ensuremath{\left\|{#1}\right\|}}

\newcommand{\trinorm}[1]{{\left\vert\kern-0.15ex\left\vert\kern-0.15ex\left\vert #1 
    \right\vert\kern-0.15ex\right\vert\kern-0.15ex\right\vert}}
\newcommand{\pare}[1]{\left({}#1\right)}
\newcommand{\p}[1]{\left({}#1\right)}
\newcommand{\curly}[1]{\left\{{}#1\right\}}
\newcommand{\abs}[1]{\ensuremath{\left\lvert{#1}\right\rvert}}

\newcommand{\inv}{^{-1}}

\def\div{\nabla \cdot}
\newcommand{\divergence}{\nabla \cdot}
\newcommand{\Grad}{\ensuremath{\nabla}}

\newcommand{\Rey}{\ensuremath{\mathrm{Re}}}
\newcommand{\ohm}{ \Omega }

\newcommand{\bfe}{\ensuremath{\mathbf{e}}}
\newcommand{\bff}{\ensuremath{\mathbf{f}}}

\newcommand{\bfx}{\ensuremath{\mathbf{x}}}

\newcommand{\bfu}{\ensuremath{\mathbf{u}}}
\newcommand{\bfv}{\ensuremath{\mathbf{v}}}

\newcommand{\bfw}{\ensuremath{\mathbf{w}}}

\newcommand{\ud}{\ensuremath{\mathbf{u}}_d}
\newcommand{\ur}{\ensuremath{\mathbf{u}}_r}

\newcommand{\vr}{\ensuremath{\mathbf{v}_r}}

\renewcommand{\wr}{\ensuremath{\mathbf{w}}_r}
\newcommand{\delt}{\ensuremath{\Delta t}}

\newcommand{\bfphi}{\ensuremath{\boldsymbol{\phi}}}

\newcommand{\phir}{\ensuremath{\bfphi_r}}
\newcommand{\bfeta}{\ensuremath{\boldsymbol{\eta}}}

\def\bu{{\bf u}}
\def\bw{{\bf w}}

\def\uwr{{\underline {w}}_r} 
 
\def\uzero{{\underline {0}}}

\def\bX{{\bf X}}
\def\bx{{\bf x}}

\def\bv{{\bf v}}
\newcommand{\bphi}{\boldsymbol{\varphi}}

\def\Ll2{\Lambda^r_{L^2}}  
\def\Lh10{\Lambda^r_{H^1_0}}  
\def\sqrtLl2h10{\sqrt{\Ll2\Lh10}}  
\def\epl2{\varepsilon_{L^2}}
\def\eph10{\varepsilon_{H^1_0}}
\def\etal2{\eta_{L^2}}
\newcommand{\ROMError}{\varepsilon_{\text{ROM}}}
\newcommand{\ClosureError}{\varepsilon_{\text{closure}}}

\newcommand{\etalimit}{\eta_{\text{consistency}}}
\newcommand{\eplimit}{\varepsilon_{\text{consistency}}}

\def\bu{{\bf u}}

\def\bw{{\bf w}}
\def\uzero{{\underline {0}}}

\def\bX{{\bf X}}
\def\bx{{\bf x}}

\def\bv{{\bf v}}
\newcommand{\de}[2]{\frac{d #1}{d #2} }
\def\ua{{\underline a}}
\def\ub{{\underline b}}
\def\uzero{{\underline 0}}
\def\wtua{\widetilde{\underline a}}
\def\wtA{{\widetilde{A}}}

\journal{ ArXiv}

\begin{document}

\begin{frontmatter}



\title{ Verifiability and Limit Consistency of
Eddy Viscosity 
Large Eddy Simulation Reduced Order Models}

\author[VT]{Jorge Reyes\corref{cor}}
\ead{reyesj@vt.edu}
\author[VT]{Ping-Hsuan Tsai}
\ead{pinghsuan@vt.edu@}
\author[VT]{Ian Moore}
\ead{ianm9123@vt.edu}
\author[VT]{Honghu Liu}
\ead{hhliu@vt.edu}
\author[VT]{Traian Iliescu}
\ead{iliescu@vt.edu}

\cortext[cor]{Corresponding author}

\address[VT]{Department of Mathematics, Virginia Tech, Blacksburg,VA 24061, USA}

\begin{abstract}
Large eddy simulation reduced order models (LES-ROMs) are ROMs that leverage LES ideas (e.g., filtering and closure modeling) to construct accurate and efficient ROMs for convection-dominated (e.g., turbulent) flows.
Eddy viscosity (EV) ROMs (e.g., Smagorinsky ROM (S-ROM)) 
are LES-ROMs whose closure model consists of a diffusion-like operator in which the viscosity depends on the ROM velocity. 
We propose the Ladyzhenskaya ROM (L-ROM), which is a generalization of the S-ROM.
Furthermore, we prove two fundamental numerical analysis results for the new L-ROM and the classical S-ROM:
(i) We prove the verifiability of the L-ROM and S-ROM, i.e, that the ROM error is bounded (up to a constant) by the ROM closure error.
(ii) We introduce the concept of ROM limit consistency (in a discrete sense), and prove that the L-ROM and S-ROM are limit consistent, i.e., that as the ROM dimension approaches the rank of the snapshot matrix, $d$, and the ROM lengthscale goes to zero, the ROM solution converges to the 
\emph{``true solution"}, i.e., the solution of the $d$-dimensional ROM.
Finally, we illustrate numerically the verifiability and limit consistency of the new L-ROM and S-ROM in two under-resolved convection-dominated problems that display sharp gradients:
(i) the 1D Burgers equation with a small diffusion coefficient; and 
(ii) the 2D lid-driven cavity flow at Reynolds number $Re=15,000$.

\end{abstract}


\end{frontmatter}

\section{Introduction} \label{S:1}

Reduced order models (ROMs) reduce the computational cost of full order models (FOMs) (i.e., classical numerical discretizations, e.g., finite element or finite volume methods) by orders of magnitude.
Galerkin ROMs (G-ROMs) leverage the available data to construct a relatively low-dimensional basis $\{ \bphi_1, \ldots, \bphi_r \}$, and use this data-driven basis in a Galerkin framework to build a ROM whose dimension is dramatically lower than the dimension of the corresponding FOM.

These ROMs have been very successful in providing an efficient and relatively accurate numerical simulation of {\it laminar} flows, e.g., 2D flow past a circular cylinder at a low Reynolds number (e.g., $Re=100$).
For these laminar flows, ROMs can yield velocity and pressure approximations that are similar to those produced by FOMs.
The ROM dimension, however, is dramatically lower than the FOM dimension.
For example, for the 2D flow past a circular cylinder at $Re=100$, the FOM dimension is $\mathcal{O}(10^5)$, whereas the ROM dimension is $\mathcal{O}(10)$.
This dimension reduction enables ROMs to reduce the FOM computational cost by orders of magnitude.

Despite their success for laminar flows, ROMs have encountered significant challenges in the numerical simulation of {\it turbulent} (i.e., high Reynolds number) flows, which are generally central to engineering and scientific applications.
Turbulent flows are challenging to approximate because they are {\it convection-dominated, chaotic}, and {\it multiscale}.
Thus, a low-dimensional ROM basis (e.g., $r =\mathcal{O}(10)$ as in the case of the 2D flow past a circular cylinder at $Re=100$) is not enough for an accurate representation of the turbulent flow dynamics.
What is needed is a much higher-dimensional ROM basis, e.g., $r =\mathcal{O}(10^3)$ or even higher~\cite{ahmed2021closures}.
These high-dimensional ROMs, however, would be impractical because their computational cost could be even higher than the FOM cost \cite{tsai2023accelerating}.
Thus, in practical applications, ROMs are used in the {\it under-resolved regime}:
To ensure a low computational cost, low-dimensional ROMs are utilized. 
Unfortunately, these low-dimensional ROMs generally yield inaccurate results, usually in the form of spurious numerical oscillations.

To tackle this inaccurate behavior in convection-dominated (e.g., turbulent) flows, under-resolved ROMs generally employ two types of strategies, which often overlap:
(i) {\it ROM stabilizations}, or 
(ii) {\it ROM closures}.  
ROM stabilizations 
(see \cite{parish2024residual} for a review) increase the ROM numerical stability and mitigate/eliminate the spurious numerical oscillations that standard ROMs yield in the under-resolved regime.
ROM closures
(see \cite{ahmed2021closures} for a review) add correction terms to the standard ROMs to increase their accuracy.

There is a wide spectrum of ROM closures, reviewed in \cite{ahmed2021closures}.
The {\it eddy viscosity (EV)} ROM closures yield some of the most successful ROMs in the under-resolved simulation of realistic turbulent flows.
The EV ROM closures leverage the fact that the main role of the discarded ROM modes in under-resolved simulations is to dissipate energy from the system \cite{CSB03}, and add numerical dissipation to the standard ROM \cite[Section IV.A]{ahmed2021closures}. 
Thus, there is a close connection between EV ROM closures and ROM stabilizations because both aim at increasing the numerical stability of the ROM (albeit by different strategies).

Just as in full order modeling \cite{sagaut2006large}, the {\it Smagorinsky ROM (S-ROM)} is one of the most popular EV ROM closures \cite{wang2012proper}.
The S-ROM adds artificial viscosity that is proportional to the Frobenius norm of the deformation tensor of the ROM velocity.
Thus, the S-ROM adds a large amount of artificial viscosity in the flow regions displaying large velocity gradients, and a negligible amount of artificial viscosity in the rest of the domain. 

One of the main contributions of this paper is the introduction of a new EV-ROM, the {\it Ladyzhenskaya ROM (L-ROM)}. 
The L-ROM is a generalization of the S-ROM that uses different powers of the Frobenius norm of the deformation tensor of the ROM velocity
and ROM lengthscale,  
which aims to allow a better treatment of sharp gradients in convection-dominated flows.  
As a first step in the assessment of the new L-ROM, we investigate it, together with the S-ROM, in the numerical simulation of two convection-dominated problems:
(i) the 1D Burgers equation with a sharp internal layer and low diffusion parameter; and
(ii) the 2D lid-driven cavity flow at Reynolds number $Re=15,000$.
We emphasize that both problems display sharp gradients that are difficult to represent with standard ROMs in the under-resolved regime. 

Another major contribution of this paper is performing a careful numerical analysis of both the new L-ROM and the classical S-ROM.
Specifically, we prove for the first time the {\it verifiability} of the L-ROM and S-ROM, i.e., we prove that the ROM error is bounded (up to a constant) by the ROM closure error.
Furthermore, we prove that the L-ROM and S-ROM are {\it limit consistent} (in a discrete sense), i.e., we prove that as the ROM dimension approaches the rank of the snapshot matrix, $ d$, and the 
ROM lengthscale goes to zero, the ROM solution converges to the
\emph{``true solution"}, i.e., the solution of the $d$-dimensional G-ROM. 
We emphasize that, to our knowledge, this is the first time that the concept of limit consistency is used at a ROM level.
We also note that, to our knowledge, verifiability has been proved for only one other ROM closure model~\cite{koc2022verifiability} (a data-driven ROM closure that is not of EV type). 

We
note that, in order to prove the verifiability and limit consistency of the L-ROM and S-ROM, we cast these two ROMs into a {\it large eddy simulation ROM (LES-ROM)}  framework (see \cite{Quaini2024Bridging} for a recent review of LES-ROMs).
In particular, we use ROM spatial filtering 
(chosen to be the ROM projection here) 
 to determine the ROM subfilter-scale stress tensor that the L-ROM and S-ROM closures aim to approximate. 
The ROM subfilter-scale stress tensor allows us to quantify the accuracy of the L-ROM and S-ROM closures, which is critical in proving the ROM verifiability.
We also note that the numerical investigation of the verifiability and limit consistency of the L-ROM and S-ROM consists of challenging under-resolved simulations of convection-dominated problems that display sharp gradients. 
 
The rest of the paper is organized as follows: 
In Section \ref{sec: Prelim_Not}, we include preliminaries and necessary background results needed for the development and numerical analysis of LES-ROMs.
In Section \ref{sec: Verifiability}, we prove the verifiability of the L-ROM and S-ROM.
In Section \ref{sec: Limit_Con}, we introduce the concept of limit consistency (in a discrete sense) in reduced order modeling.
Furthermore, we prove that both the L-ROM and S-ROM are limit consistent.
In Section \ref{sec-numres}, we investigate the verifiability and limit consistency of the L-ROM and S-ROM in the under-resolved numerical simulation of two convection-dominated problems:
(i) the 1D Burger equation with a small diffusion coefficient; and 
(ii) the 2D lid-driven cavity flow at Reynolds number $Re=15,000$.
In Section \ref{sec-Con}, we present the conclusions of our theoretical and numerical investigation, and outline research directions for future work.

\section{Notations and Preliminaries} \label{sec: Prelim_Not}

As a mathematical model, we consider the incompressible Navier-Stokes equations (NSE) given as 
\begin{align}
\bfu_t+ (\bfu \cdot \nabla) \bfu -\Rey\inv \Delta \bfu+\nabla p &=\bff
\qquad
\text {in } \ohm \times (0,T],\label{eq:Strong_NSE1}\\
\divergence \bfu&=0
\qquad
\text {in } \ohm \times (0,T],\label{eq:Strong_NSE2}
\end{align}
where $\bfu$ 
is the fluid velocity, 
$p$ 
the pressure, 
$\ohm$ the spatial domain, 
$[0,T]$ the time interval, 
$\bff$ 
the external force, and $ \Rey$ 
the Reynolds number. 
For simplicity, we consider homogeneous Dirichlet boundary conditions.
Appropriate 
initial conditions are needed to close the system.

We will use the following standard function spaces: $L^p(\Omega) $, $W^{k,p}(\Omega) $, and $ H^k(\Omega) = W^{k,2}(\Omega)$, where $ k \in \mathbb{N}$ and $1 \leq p \leq \infty $, for a bounded domain $\Omega \subset\mathbb{R}^\ell $ $(\ell=2,3)$.
 The $L^2(\Omega)$ norm is denoted as $\Vert \cdot \Vert$, with the corresponding inner product $(\cdot,\cdot)$. The $L^p(\Omega)$ norm is denoted by $\| \cdot \|_{p} $, while the Sobolev $ W^{k,p}(\Omega)$ norm is $ \Vert \cdot \Vert_{k,p} $. 

The solutions of the continuous NSE are traditionally sought in the functional spaces $ X= W_{0}^{1,2}(\ohm)$ and $Q= L_0^2(\ohm) = \curly{ q\in L^2 \mid \int_{\ohm} q d\ohm =0}$ for velocity and pressure, respectively.  

In the case of a continuous LES model, it is common to impose an additional regularity assumption on the velocity~\cite{berselli2006mathematics}. 
For the Smagorinsky and Ladyzhenskaya models, we consider $ X_s = W_0^{1,s+2}$

for a fixed parameter $s$. It is clear that when $s=0$ we have $ X_0 = X$. For the Ladyzhenskaya model in \cite{ladyzhenskaya1985boundary}, Ladyzhenskaya showed that solutions to her model are globally unique in time for any Reynolds number and any $ s \geq \frac12 $. This result was then improved by Du and Gunzburger in \cite{du1991analysis} to any $s \geq \frac15$. 

In this paper, we use the Ladyzhenskaya and Smagorinsky models only at the (discrete) ROM level 

and consider $s \geq 0$.
 
We also let $\delt$ denote the time step, 
$t^{n}
= n \delt$, $n = 0, 1, \dots, M$, the time instances where the ROM approximation is computed, and 
$T := M \delt$ the final time. 
We then define the following discrete norms:
\begin{eqnarray} \label{eq: triplebar}
 \trinorm{ \bfu }_{\infty, k} := \max_{0 \le n \le M} \|\bfu^{n}\|_{k,2} \,, \qquad
\trinorm{ \bfu }_{m,k} := \left(\sum_{n=0}^{M} \| \bfu^{n} \|^{m}_{k,2}\delt\right)^{1/m}.
\end{eqnarray}
For $ \bfu,\bfv,\bfw \in X_s $, we define the trilinear forms $ b, b^*: X_s \times X_s \times X_s 
\longrightarrow \mathbb{R}  $ by
\begin{eqnarray*}
 b(\bfu,\bfv,\bfw) &=& (\bfu\cdot\nabla \bfv,\bfw),\\
 b^{*}(\bfu,\bfv,\bfw)&=& (\bfu \cdot \Grad\bfv ,\bfw) + \frac{1}{2}((\Grad\cdot \bfu)\bfv,\bfw) = \frac{1}{2}(b(\bfu,\bfv,\bfw)-b(\bfu,\bfw,\bfv)).
 \end{eqnarray*}

We also use the following lemmas: 
\begin{lemma}
[\cite{layton2008introduction,temam2001navier}]
\label{TRIL} 
For $\bfu, \bfv, \bfw \in X_s $, the trilinear term $b^{*}(\bfu, \bfv, \bfw)$ can be bounded as follows:
\bea
 b^{*}(\bfu, \bfv, \bfw) & \leq &  C(\ohm) \left\|\bfu\right\|^{\frac{1}{2}} \left\|\Grad \bfu\right\|^{\frac{1}{2}} \left\|\Grad \bfv\right\|  \left\|\Grad \bfw\right\|, \\
 b^{*}(\bfu, \bfv, \bfw) & \leq & C(\ohm) \left\|\Grad \bfu\right\| \left\|\Grad \bfv\right\|  \left\|\Grad \bfw \right\|.
\eea
\end{lemma}

\begin{lemma}[Strong monotonicity \cite{du1991analysis,minty1962monotone,lions1969quelques}] 
\label{lemma: Strong Mono}
For $ \bfu, \bfv  \in W^{1,s+2}(\Omega)$, there exists a 
positive constant 
$C$ depending on $\ell,s$, and $ \Omega$, such that the following inequality holds: 
\begin{gather}
\label{StrongMonotonicty}
    \left( \| \Grad \bfu  \|_F^{s} \Grad \bfu - \| \Grad \bfv  \|_F^{s} \Grad \bfv , \Grad(\bfu - \bfv)  \right) \geq C \| \Grad (\bfu - \bfv)  \|_{s+2}^{s+2},
\end{gather}
 where $\norm{\mathbf{A}}_F = \pare{\sum_{i,j=1}^{\ell}{a_{i,j}^2}}^{1/2}$ denotes the Frobenius norm of a square matrix $\mathbf{A}\in\mathbb{R}^{\ell \times \ell}$. 
\end{lemma}

\begin{lemma}[Discrete Gronwall Lemma \cite{heywood1990finite}]
\label{lemma: discreteGronwall} Let $\Delta t$, H, and $a_{n},b_{n},c_{n},d_{n}$
(for integers $n \ge 0$) be finite nonnegative numbers such that
\begin{equation}
a_{l}+\Delta t \sum_{n=0}^{l} b_{n} \le \Delta t \sum_{n=0}^{l} d_{n}a_{n} +
\Delta t\sum_{n=0}^{l}c_{n} + H \ \ \mathrm{for} \ \ l\ge 0. \label{gronwall1}
\end{equation}

Suppose that $\Delta t \, d_n < 1 \; \forall n$. Then,
\begin{equation}
a_{l}+ \Delta t\sum_{n=0}^{l}b_{n} \le \exp\left( \Delta t\sum_{n=0}^{l} \frac{d_{n}}{1 - \Delta t d_n } \right) \left( \Delta t\sum_{n=0}^{l}c_{n} + H
\right)\ \ \mathrm{for} \ \ l \ge 0.
\end{equation}
\end{lemma}

\subsection{LES background}\label{ssec: LES_back}

Numerical simulations of turbulent flows generally occur in the {\it under-resolved regime}, that is, when not enough resolution is available to represent the underlying complex dynamics down to the Kolmogorov scale.
{\it Large eddy simulation (LES)}  \cite{berselli2006mathematics,john2003large,sagaut2006large} aims at increasing the accuracy of these under-resolved simulations by adding correction terms, that is, models for the effect of the discarded scales.
In LES, the large, resolved scales and the small, unresolved scales are defined by using a spatial filter. 
Applying 
 
a spatial filter denoted with $ \overline{\, \cdot \,}$, the NSE \eqref{eq:Strong_NSE1}-\eqref{eq:Strong_NSE2} become the spatially filtered NSE,
given as
\begin{align}
\overline{\bfu}_t+ \divergence(\overline{\bfu~\bfu}) -\Rey\inv \Delta \overline{\bfu}+\nabla \overline{p} &=\overline{\bff},\label{eq:Strong_SFNSE1}\\
\divergence \overline{\bfu}&=0.\label{eq:Strong_SFNSE2}
\end{align}

We note that the system \eqref{eq:Strong_SFNSE1}-\eqref{eq:Strong_SFNSE2} is not \textit{closed} for the filtered velocity $ \overline{\bfu}$ due to $ \overline{\bfu~ \bfu} \neq \overline{\bfu}~ \overline{\bfu}$. 

Thus, the spatially filtered NSE \eqref{eq:Strong_SFNSE1}-\eqref{eq:Strong_SFNSE2} are commonly 
rewritten as 
\begin{align}
\overline{\bfu}_t+ \divergence(\overline{\bfu}~\overline{\bfu}) -\Rey\inv \Delta \overline{\bfu}+ \divergence \uptau(\bfu)+\nabla \overline{p} &=\overline{\bff},\label{eq:Strong_SFNSE3}\\
\divergence \overline{\bfu}&=0,\label{eq:Strong_SFNSE4}
\end{align}
where $ \uptau(\bfu) := \overline{\bfu ~ \bfu} -\overline{\bfu}~\overline{\bfu}$ is called the 
subfilter-scale stress tensor. 

Closure models then seek to approximate the true closure $\uptau(\bfu)$ with a computationally more efficient tensor that depends only on the filtered velocity $ \uptau^{\text{LES}}(\overline{\bfu})$. A perfect closure model has not yet been found, hence in general $ \uptau(\bfu) \neq \uptau^{\text{LES}}(\overline{\bfu})$. 
Thus, substituting $\uptau(\bfu)$ with $\uptau^{\text{LES}}(\overline{\bfu})$ in \eqref{eq:Strong_SFNSE3}, the solution to this new system can only be an approximation of $\overline{\bfu}$ at best.

The class of closure models that attempt to account for the energy lost due to the filtered small eddies by adding additional diffusion are known as eddy viscosity (EV) models 
\cite{sagaut2006large,berselli2006mathematics}. 
Classical examples of 
EV models include the Smagorinsky model \cite{smagorinsky1963general,lilly1967representation}, introduced in 1963, and the Ladyzhenskaya model \cite{ladyzhenskaya1985boundary,ladyzhenskaya1969mathematical}, proposed 
in the context of regularization. 
Other well-known 
EV closure models include the dynamic Smagorinsky model
\cite{GPMC91} and the variational multiscale (VMS) model \cite{HughesVMS1998}. 
LES is a highly developed field, with closure models of various types (surveyed, e.g., in~\cite{sagaut2006large}).

Denoting the approximations to $(\overline{\bfu},\overline{p})$ in \eqref{eq:Strong_SFNSE3}-\eqref{eq:Strong_SFNSE4} as $(\bfw,q)$ leads to the LES equations 
\begin{align}
{\bfw}_t+ \divergence({\bfw}~{\bfw}) -\Rey\inv \Delta {\bfw}+ \divergence \uptau^{\text{LES}}(\bfw)+\nabla {q} &= {\bff},\label{eq:Strong_LES1}\\
\divergence {\bfw}&=0.\label{eq:Strong_LES2}
\end{align}

This naturally leads to the question of 
quantifying and analyzing the LES modeling error,
$ \norm{ \overline{\bfu} - \bfw }$. 
In LES, these concepts were 
pioneered by Layton and his group  
(see, e.g., \cite{Kaya2002verifiability} and the review in \cite{berselli2006mathematics}).

\subsection{Galerkin ROM background}

In this section, we introduce the 
{Galerkin ROM (G-ROM)}. 
{To this end, we collect data as $M$ snapshots 
of FOM solutions, that is, $\curly{ \bfu_{\text{FOM}}^0, \dots, \bfu_{\text{FOM}}^M }$, where $\bfu^n_{\text{FOM}} = \bfu_{\text{FOM}}(t^n), 
\, n = 0, \ldots, M$. We follow the standard proper orthogonal decomposition (POD) procedure \cite{berkooz1993proper,volkwein2013proper} to construct the reduced basis functions.}
POD 
seeks a 
basis $\{\bphi_1,\ldots,\bphi_r\}$ in $L^2$ that optimally approximates the snapshots, that is, solves the following minimization problem: 
\begin{equation*}
    \min_{
    \{ \bphi_j \}_{j=1}^{r}} \ \frac{1}{M+1}\sum^M_{n=0}\norm{\bfu_{\text{FOM}}(\cdot,t^n)- \sum^r_{j=1}\left(\bfu_{\text{FOM}}(\cdot,t^n),\bphi_j(\cdot)\right)\bphi_j(\cdot)}^2, 
\end{equation*}
subject to the conditions $(\bphi_i,\bphi_j)
= \delta_{ij}$, for $1\le i,j \le r$, where $\delta_{ij}$ is the Kronecker delta. The minimization problem can be solved 
by considering the eigenvalue
problem $\mathcal{K}\underline{z}_j = \lambda_j \underline{z}_j$, for $j=1,\ldots,r$,
where $\mathcal{K} \in \mathbb{R}^{(M+1)\times(M+1)}$ is the snapshot Gramian matrix using the $L^2$ inner product (see, e.g.,~\cite{iollo2000stability,kaneko2020towards,
tsai2022parametric} for alternative strategies).
Denoting by $d$ the rank of the Gramian matrix, $\mathcal{K}$,  
the $d$-dimensional G-ROM is constructed by inserting the ROM basis expansion
\bea \label{eq: ud}
    \bfu_d (\bfx,t) =  \sum_{j=1}^{d} u_{d,j}(t)\bphi_j(\bfx)
\eea
into the weak form of the NSE \eqref{eq:Strong_NSE1}--\eqref{eq:Strong_NSE2}, 
followed by backward Euler discretization in time: 
For $ n=1, \dots, M$, find $\ud^n \in \bX_d$ such that, for all    $\bv_d \in \bX_d$,
\begin{eqnarray}
    && 
    \left(
        \frac{ \bu_d^{n} - \bu_d^{n-1} }{\delt} , \bv_d 
    \right)
    + \Rey^{-1} \, 
    \left( 
        \nabla \bu_d^n , 
        \nabla \bv_d
    \right)
    + b^*(\ud^n, \ud^n , \vr)
    = \pare{ \bff^n, \bv_d},\qquad
    \label{eq: gromu}
\end{eqnarray}
where $\bX_d:=\mathrm{span}\{\bphi_1,\cdots, \bphi_d\}$ is the $d$-dimensional ROM space, 
{
which spans the same subspace as the snapshot data.}

In 
engineering and scientific applications (e.g., turbulent flow simulations), the $d$-dimensional G-ROM is seldom used due to its relatively high computational cost \cite{tsai2023accelerating}. 
Instead, a further reduced dimension $ r \ll d$ is chosen,
and a lower-dimensional space $ \bX_r :=  \text{span} \curly{\bphi_1,\ldots,\bphi_r}$ is employed to construct an $r$-dimensional G-ROM: For $ n=1, \dots, M$, find $\ur^n \in \bX_r$ such that, for all    $\bv_r \in \bX_r$,

\begin{eqnarray}
    && 
    \left(
        \frac{\bu_r^n - \ur^{n-1}}{ \delt} , \bv_r 
    \right)
    + \Rey^{-1} \, 
    \left( 
        \nabla \bu_r^n , 
        \nabla \bv_r 
    \right)
    + b^*( \ur^n, \ur^n ,\vr)
    = \pare{\bff^n,\vr}. \qquad
    \label{eq: gromu-r}
\end{eqnarray}

\begin{remark} 
    Because the POD basis functions are a linear combination of the FOM generated snapshots, they satisfy the boundary conditions of the original PDE and inherit the FOM's weakly divergence-free 
    property. 
    Because the ROM velocity is 
    only weakly divergence-free, to 
    ensure stability, we equip the ROM with the skew-symmetric trilinear form $b^*$~\cite{layton2008introduction}. 
\end{remark}

\begin{definition}[ROM $L^2$ projection \cite{
KV01}]\label{def: ROM_Proj}
Let $ P_r: L^2(\ohm) \to \bX_r $ be the $ L^2$ projection onto $ \bX_r$ such that, $ \forall \bfu \in L^2(\ohm)$, $ P_r(\bfu)$ is the unique element of $ \bX_r$ such that
\begin{equation}
    (P_r(\bfu),\vr) = (\bfu,\vr), \qquad \forall \vr\in \bX_r.
    \label{eqn:rom-projection-definition}
\end{equation}
\end{definition}

\begin{definition}[Generic Constant C]
We denote with $C$ 
generic constants that do not depend on $r$, $ \delt$, $\Rey$, or the ROM lengthscale parameter, $\delta$, given in \eqref{eq: LZ_closure} below, but can depend on 
$\ohm$, $\bff$, and $\bu_d$. 
\end{definition} 

In Lemma \ref{lemma: Stability}, we present the unconditional stability of the traditional G-ROM.
The proof follows along the lines of 
the stability proofs 
of more complex models, as can be seen in, e.g.,   \cite{iliescu2014variational,xie2018numerical,reyes2024trrom}.  
\begin{lemma}\label{lemma: Stability}
For any fixed $ r = 1, \dots, d$, the solution 
of the G-ROM 
\eqref{eq: gromu-r} is unconditionally stable: 
For any $\Delta t>0$, the solution satisfies:
\begin{equation}\label{eq: stab}
   \trinorm{\ur}^2_{\infty,0}  + \Rey\inv |||\nabla \ur|||_{2,0}^2 \leq ||\ud^0||^2 
   + \Rey \, \delt \sum_{n=0}^{M}||\bff^{n}||_{-1}^2 := \Rey \  C.
\end{equation}
\end{lemma}

\subsection{LES-ROMs background}
\label{ssec: les_rom}

As mentioned in Section~\ref{S:1}, in convection-dominated (e.g., turbulent) flows, under-resolved ROMs generally yield inaccurate results, often in the form of spurious oscillations.
One popular strategy for alleviating this inaccurate behavior is to endow the ROMs with closure models.
In this paper, we analyze and investigate numerically two EV ROM closure models.
To perform the numerical analysis of these EV ROM closures, we cast them in an LES-ROM framework, which we briefly overview in this section.

LES-ROMs, which have been developed over the past decade (see the reviews in \cite{Quaini2024Bridging,ahmed2021closures}) adapt concepts from classical LES \cite{berselli2006mathematics,sagaut2006large} to construct ROM closures.
In particular, both LES-ROMs and classical LES models (i.e., LES-FOMs) leverage spatial filtering to separate the large and small flow scales at coarse resolutions - using coarse meshes in LES-FOMs and a limited number of basis functions in LES-ROMs.
Furthermore, both LES-ROMs and LES-FOMs approximate the large scales in the flow, and model the action of the small scales.

We emphasize, however, that LES-ROMs and classical LES-FOMs are fundamentally different in the type of spatial filtering used to construct the closure models:
Classical LES-FOMs often employ continuous spatial filters, e.g., the Gaussian or differential filters \cite{berselli2006mathematics,sagaut2006large}.
Specifically, the construction of LES-FOMs generally is a three-step process:
(i) filter the NSE~\eqref{eq:Strong_NSE1}--\eqref{eq:Strong_NSE2}; 
(ii) model the subfilter-scale stress tensor in the resulting spatially filtered NSE~\eqref{eq:Strong_SFNSE1}--\eqref{eq:Strong_SFNSE2}; and 
(iii) use standard numerical methods (e.g., the finite element method) to discretize the resulting LES-FOM~\eqref{eq:Strong_LES1}--\eqref{eq:Strong_LES2}. 
See \cite{agdestein2025discretize} for recent discussion on alternatives. 

In contrast, LES-ROMs use {\it discrete} spatial filters, e.g., the ROM $L^2$ projection, the ROM differential filter, or the ROM higher-order algebraic filter~\cite{tsai2025time}.
The reason for using a discrete (ROM) spatial filter is that the LES-ROM construction is 
generally a two-step process that starts with a discrete set of equations:
(i) filter the (discrete) $d$-dimensional G-ROM~\eqref{eq: gromu}, whose solution represents our \emph{``true solution"} \cite{koc2022verifiability}; and 
(ii) construct a closure model for the subfilter-scale stress tensor in the resulting spatially filtered equations. 
Next, we outline the main steps in the LES-ROM construction.

As our (discrete) ROM spatial filter, we use $P_r$, the ROM projection in Definition~\ref{def: ROM_Proj}.
Filtering the (discrete) $d$-dimensional G-ROM \eqref{eq: gromu} with $P_r$, we obtain the spatially-filtered $d$-dimensional G-ROM, which is the (discrete) ROM equivalent of the (continuous) spatially filtered NSE \eqref{eq:Strong_SFNSE1}-\eqref{eq:Strong_SFNSE2} used in classical LES-FOM construction:

For $ n=1, \dots, M$, find $P_r(\ud^n) \in \bX_r$ such that, for all $\bv_r \in \bX_r$,
\begin{eqnarray} 
   \frac{1}{\delt}\pare{ P_r(\bfu_d^{n}) - P_r(\bfu_d^{n-1}) , \vr } + \Rey\inv (\nabla P_r(\bfu_d^{n}) , \nabla \vr ) + \Rey\inv \pare{ \mathcal{E}^n(\bu_d^n),\Grad \vr} \nonumber \\
    + b^*(P_r(\bfu_d^{n}),P_r(\bfu_d^{n}),\vr) + \pare{ \uptau^{\mathrm{FOM}} (\bfu_d^{n}), \vr  } = (\bff(t^{n}),\vr), \qquad \label{eq: Filtered-ROM}
\end{eqnarray}
where the 
closure term, 
$\uptau^{\mathrm{FOM}}$, and the commutation error, $\mathcal{E}^n$, are defined as follows:  

\al
    \uptau^{\mathrm{FOM}}(\bfu_d^{n}) &:= P_r( \bfu_d^n \cdot \Grad \bfu_d^n  ) - P_r(\bfu_d^n) \cdot \Grad P_r(\bfu_d^n), \label{eq: tau_FOM}\\
    \mathcal{E}^n (\bu_d^n) &:= P_r(\Grad \bfu_d^n) - \Grad P_r(\bfu_d^n).\label{eq: commutation_error}
\eal
The commutation error $\mathcal{E}^n$ 
was shown in \cite{koc2019commutation} to be generally nonzero, but 
negligible for large $\Rey$. 
Because this paper focuses on under-resolved ROMs of convection-dominated flows, we follow the approach in \cite{koc2022verifiability} and 
neglect the commutation error.

We note that the spatially filtered $d$-dimensional G-ROM \eqref{eq: Filtered-ROM} suffers 
from 
a closure problem that is similar to the closure problem in the classical LES-FOM case (see \eqref{eq:Strong_SFNSE1}--\eqref{eq:Strong_SFNSE2}) because the 
subfilter-scale stress tensor, $\uptau^{\mathrm{FOM}}(\bfu_d^{n})$ in \eqref{eq: tau_FOM}, is not closed in terms of $P_r(\bfu_d^n)$. 
Thus, just as in the LES-FOM construction (Section~\ref{ssec: LES_back}), we model the subfilter-scale stress tensor, $\uptau^{\mathrm{FOM}}$, with a ROM closure model, $\uptau^{\mathrm{ROM}}$, to obtain the LES-ROM: For $ n=1, \dots, M$, find $\wr^n \in \bX_r$ such that

\begin{eqnarray} \label{eq: Les-ROM}
\begin{split}
   \frac{1}{\delt}\pare{\wr^{n} - \wr^{n-1} , \vr } + \Rey^{-1}(\nabla \wr^{n}, \nabla \vr)+ b^*(\wr^{n},\wr^{n},\vr)\\
    + \pare{ \uptau^{\mathrm{ROM}} (\wr^{n}), \vr } = (\bff(t^{n}),\vr),
\end{split} \quad \forall \vr\in \bX_r, \quad 
\end{eqnarray}

where $\wr^n$ is the LES-ROM approximation of $P_r(\bfu^n_d)$.

There is a wide variety of LES-ROMs, which are reviewed in \cite{ahmed2021closures,Quaini2024Bridging}.
In the next section, we 
prove the verifiability of a class of LES-ROMs and illustrate this strategy for two EV LES-ROMs: the Smagorinsky ROM and the new Ladyzhenskaya ROM.

\begin{remark}

To build the LES-ROM framework, we have used $P_r$, the ROM projection in Definition~\ref{def: ROM_Proj}.
Given the intrinsic hierarchical structure of the POD basis, the ROM projection has been popular in the construction of LES-ROMs (see, however, \cite{koc2019commutation,xie2017approximate,Quaini2024Bridging} for the use of ROM differential filters to construct LES-ROMs).
The ROM projection was used for the first time to develop the dynamic subgrid-scale ROM closure model~\cite{wang2012proper}.
It was later used 
for variational multiscale (VMS) ROM closures (e.g., \cite{bergmann2009enablers,iliescu2014variational,reyes2020projection,stabile2019reduced, wang2012proper}).
Recently, the ROM projection was leveraged 
to construct data-driven VMS-ROMs~\cite{koc2022verifiability,manti2025symbolic}.
\end{remark}

\begin{remark} 
In the numerical analysis of the LES-ROM framework developed in the next section, for simplicity, we consider the backward Euler time discretization for all ROMs. 
We believe, however, that these results can be extended to higher-order time discretizations as long as both the G-ROM and LES-ROM use the same time discretization (i.e., are consistent \cite{ingimarson2022full,strazzullo2022consistency} with respect to time discretization).
\end{remark}

\section{Verifiability of LES-ROMs } \label{sec: Verifiability}

The first goal of this section is to prove the verifiability of a certain class of LES-ROMs (see Theorem~\ref{thm:verifiability}).
The second goal is to leverage Theorem~\ref{thm:verifiability} to prove the verifiability of two EV LES-ROMs: the new Ladyzhenskaya ROM and the Smagorinsky ROM (see Corollary \ref{cor: verifiability} and Corollary \ref{cor: s_verifiability}, respectively).
Verifiability of closure models has been investigated for decades in classical CFD (see, e.g., \cite{berselli2006mathematics} and \cite{Kaya2002verifiability} for a survey of verifiability methods in LES).
In reduced order modeling, however, the study of verifiability for LES-ROMs has just begun.
Indeed, verifiability has been investigated only in \cite{koc2022verifiability} for a particular data-driven LES-ROM. 
In this paper, we continue the effort started in \cite{koc2022verifiability} and extend the verifiability theory to a larger class of LES-ROMs.
Furthermore, we prove verifiability for two new LES-ROMs of eddy viscosity type: the Smagorinsky and Ladyzhenskaya models.

First, we define verifiability and mean dissipativity following \cite{koc2022verifiability} (see also \cite{Kaya2002verifiability}).

\begin{definition}[Verifiability \cite{koc2022verifiability}]
    For a fixed number of snapshots, M, a ROM closure model is verifiable in the $L^2 $ norm 
    if there exists a constant $C$ such that, for all $ r \leq d$ and $ n = 1, \dots, M$, the following a priori error bound holds:
    \begin{equation}\label{eq: Verify}
        \norm{ P_r(\bfu_d^n) - \wr^n }^2 \leq C \frac{1}{M}\sum_{j=1}^{M} \norm{P_r(\uptau^{\mathrm{FOM}}(\bfu_d^j)) - \uptau^{\mathrm{ROM}}(P_r(\bfu_d^j))}^2. 
    \end{equation}
\end{definition}

\begin{definition}[Mean Dissipativity \cite{koc2022verifiability}]\label{def: mean_d}
    A ROM closure model satisfies the mean dissipativity condition if for all $\ur, \vr \in \bX_r $,
    \begin{equation}\label{eq: mean_dis}
        0 \leq ( \uptau^{\mathrm{ROM}}(\ur) - \uptau^{\mathrm{ROM}}(\vr), \ur - \vr ) .
    \end{equation}
\end{definition}

In  Theorem~\ref{thm:verifiability}, 
we prove that all LES ROMs \eqref{eq: Les-ROM} are verifiable if the closure model, $ \uptau^{\mathrm{ROM}}$,  satisfies the mean dissipativity condition and the time step $\delt$ is sufficiently small. 

\begin{theorem}
    The LES-ROM \eqref{eq: Les-ROM} that satisfies the \textit{mean dissipativity condition} \eqref{eq: mean_dis}  
    with the initial condition $ \wr^0 = P_r(\bfu_d^0)$ is verifiable: 
    If the time step $\delt$ is chosen such that $ \delt \, d_n \leq 1$ for all $ n= 1, \dots ,M$, then 
    \begin{equation}
        \trinorm{\bfe}_{\infty,0}^2 + \Rey\inv \trinorm{\Grad \bfe}_{2,0}^2 \leq K \trinorm{ P_r( \uptau^{\mathrm{FOM}}\pare{\bfu_d} - \uptau^{\mathrm{ROM}}\pare{P_r\pare{\bfu_d}} ) }_{2,0}^2, \label{eq:verifiability}
    \end{equation}
    where $ K = \exp\p{\delt \sum_{n=1}^{M}\p{ \frac{d_n}{1-\delt d_n} }}$, $ d_n = C \p{\Rey^{3} \norm{\Grad P_r \p{\bfu_d^{n}} }^4 + 1} $, and $ \bfe = P_r(\bfu_d) - \wr$. 
    \label{thm:verifiability}
\end{theorem}

\begin{proof}
    We subtract \eqref{eq: Les-ROM} from \eqref{eq: Filtered-ROM} and then add 
    $ ( \uptau^{ROM}(P_r(\bfu_d^n)), \vr)$ to both sides to get the error equation 
    \bea
    & &\frac{1}{\delt} \pare{\bfe^{n} - \bfe^{n-1},\vr} + \Rey\inv \pare{\Grad \bfe^{n}, \Grad \vr} + b^*\p{P_r\p{\bfu_d^{n}},P_r\p{\bfu_d^{n}},\vr} \nonumber\\
    & &- b^*\p{\wr^{n},\wr^{n},\vr} + \p{ \uptau^{ROM}(P_r(\bfu_d^n))- \uptau^{\mathrm{ROM}}(
    \wr^n), \vr } \nonumber\\
    &=& - \pare{\uptau^{\mathrm{FOM}}\pare{\bfu_d^{n}},\vr} +  \p{ \uptau^{\mathrm{ROM}}(P_r(\bfu_d^n)), \vr},
    \eea
    where $ \bfe^n = P_r(\bfu_d^n) - \wr^n$. 
    Note that the nonlinear terms can be rewritten as
    \begin{eqnarray*}
        b^{\ast}( P_r(\bfu_d^{n}), P_r(\bfu_d^{n}),\vr) - b^{\ast}( \wr^{n}, \wr^{n},\vr)
        = b^{\ast}(  \bfe^{n}, P_r(\bfu_d^{n}),\vr)  
        + b^{\ast}( \wr^{n},  \bfe^{n},\vr) .
    \end{eqnarray*}
    Using the above equality, 
    letting $\vr=  \bfe^{n}$, and noting that 
        $b^{\ast}( \wr^{n}, \bfe^{n}, \bfe^{n})=0$ because $ b^*$ is skew-symmetric with respect to its second and third arguments, we obtain 
   \bea
    & &\frac{1}{\delt} \pare{\bfe^{n} - \bfe^{n-1},\bfe^{n}} + \Rey\inv \norm{\Grad \bfe^{n}}^2 + b^*\p{\bfe^{n},P_r\p{\bfu_d^{n}},\bfe^{n}} \nonumber\\
    & &+ \p{ \uptau^{\mathrm{ROM}}(P_r(\bfu_d^n)) - \uptau^{\mathrm{ROM}}(
    \wr^n),\bfe^n } \nonumber \\
    &=& - \pare{\uptau^{\mathrm{FOM}}\pare{\bfu_d^{n}} -  \uptau^{\mathrm{ROM}}\p{P_r(\bfu_d^n)}   ,\bfe^{n}}.
    \eea
    Applying the Cauchy–Schwarz and Young's inequalities to the 
    first term results in
    \als
    &\frac{1}{2\delt} \pare{ \norm{\bfe^{n}}^2 - \norm{\bfe^{n-1}}^2 } + \Rey\inv \norm{\Grad \bfe^{n}}^2  \nonumber\\
    &+  \p{ \uptau^{\mathrm{ROM}}(P_r(\bfu_d^n)) - \uptau^{\mathrm{ROM}}(
    \wr^n) ,\bfe^n }  \nonumber \\
    \leq&  \abs{b^*\p{\bfe^{n},P_r\p{\bfu_d^{n}},\bfe^{n}} } 
    + \abs{ \pare{\uptau^{\mathrm{FOM}}\pare{\bfu_d^{n}} - \uptau^{\mathrm{ROM}}(P_r(\bfu_d^n))  ,\bfe^{n}}}.       
    \eals
    By the mean dissipativity of the closure  
    model, we have
    \begin{equation}\label{eq: mean-diss}
         \p{ \uptau^{\mathrm{ROM}}(P_r(\bfu_d^n)) - \uptau^{\mathrm{ROM}}(
         \wr^n),\bfe^n }  \geq 0.
    \end{equation} 
    Therefore, we can drop that term from left-hand side,
    yielding, 
    \bea
    \begin{split}
    \frac{1}{2\delt} \pare{ \norm{\bfe^{n}}^2 - \norm{\bfe^{n-1}}^2 } + \Rey\inv \norm{\Grad \bfe^{n}}^2 
    \leq  \abs{b^*\p{\bfe^{n},P_r\p{\bfu_d^{n}},\bfe^{n}} } \\
    + \abs{ \uptau^{\mathrm{FOM}}\pare{\bfu_d^{n}} - \uptau^{\mathrm{ROM}}(P_r(\bfu_d^n))  ,\bfe^{n}}. \label{eq: ver1}
    \end{split}
    \eea
    We then bound the two terms on the right-hand side in standard ways. 
    For the nonlinear term, we use Lemma \ref{TRIL}:
    \bea 
    & & \abs{b^*\p{\bfe^{n},P_r\p{\bfu_d^{n}},\bfe^{n}} } \leq C \norm{\Grad \bfe^{n}}^{3/2} \norm{\Grad P_r\p{\bfu_d^{n}}} \norm{\bfe^{n}}^{1/2} \nonumber \\
    & \leq& \frac{\Rey\inv}{2} \norm{\Grad \bfe^{n}}^2 + C \Rey^{3} \norm{\Grad P_r \p{\bfu_d^{n}} }^4 \norm{\bfe^{n}}^2. \label{eq: ver_bstar}
    \eea
    For the closure term, we use \eqref{eqn:rom-projection-definition} along with the Cauchy–Schwarz and Young's inequalities, which gives
    \bea 
    \abs{ \pare{\uptau^{\mathrm{FOM}}\pare{\bfu_d^{n}}  -  \uptau^{\mathrm{ROM}}(P_r(\bfu_d^n)),\bfe^{n}}} = \nonumber \\
    \abs{ \pare{ P_r( \uptau^{\mathrm{FOM}}\pare{\bfu_d^{n}}  -  \uptau^{\mathrm{ROM}}(P_r(\bfu_d^n)) ),\bfe^{n}}} \leq \nonumber \\
    \frac{1}{2} \norm{P_r( \uptau^{\mathrm{FOM}}\pare{\bfu_d^{n}} -  \uptau^{\mathrm{ROM}}(P_r(\bfu_d^n)))}^2 + \frac{1}{2}  \norm{\bfe^{n}}^2. \label{eq: ver_close}
    \eea
    Plugging in \eqref{eq: ver_bstar} and \eqref{eq: ver_close} into \eqref{eq: ver1} results in 
    \bea
    \begin{split}
    \frac{1}{2\delt} \pare{ \norm{\bfe^{n}}^2 - \norm{\bfe^{n-1}}^2 } + \frac{\Rey\inv}{2} \norm{\Grad \bfe^{n}}^2 \\
    \leq 
    C \p{\Rey^{3} \norm{\Grad P_r \p{\bfu_d^{n}} }^4 +\frac{1}{2}} \norm{\bfe^{n}}^2\\
    + \frac{1}{2} \norm{P_r(\uptau^{\mathrm{FOM}}\pare{\bfu_d^{n}} - \uptau^{\mathrm{ROM}}(P_r(\bfu_d^n)) )}^2. 
    \label{eqn:pre-gronwall}
    \end{split}
    \eea
    Because $ \wr^0 = P_r(\bfu_d^0)$, $ \norm{\bfe^0}^2 = \norm{ P_r(\bfu_d^0) - \wr^0 }^2= 0$.
    Thus, multiplying \eqref{eqn:pre-gronwall} by $ 2\delt$ and summing from $ n=1 $ to $ M$ gives 
    \bea
    \begin{split}
    \norm{\bfe^{M}}^2 + \Rey\inv  \delt \sum_{n=1}^{M}\norm{\Grad \bfe^{n}}^2 
    \leq  
    C  \delt \sum_{n=1}^M \p{\Rey^{3} \norm{\Grad P_r \p{\bfu_d^{n}} }^4 + 1} \norm{\bfe^{n}}^2\\
    + \delt \sum_{n=1}^{M} \norm{P_r(\uptau^{\mathrm{FOM}}\pare{\bfu_d^{n}} -\uptau^{\mathrm{ROM}}\pare{P_r\pare{\bfu_d^{n}} })}^2.
    \label{eqn:gronwall}
    \end{split}
    \eea
    Choosing $ \delt$ sufficiently small, $ \delt \leq C\p{\Rey^{3} \norm{\Grad P_r \p{\bfu_d^{n}} }^4 + 1}\inv $, and applying the discrete Gronwall's inequality (Lemma \ref{lemma: discreteGronwall})  to \eqref{eqn:gronwall} gives the result.
\end{proof}

\subsection{Ladyzhenskaya and Smagorinsky LES-ROMs}

In this section, we specialize the results obtained for general LES-ROMs (i.e., Theorem~\ref{thm:verifiability}) for two eddy viscosity LES-ROMs: the Ladyzhenskaya and the Smagorinsky models.
We emphasize that this is the first time that the Ladyzhenskaya model is used in reduced order modeling.

The {\it Ladyzhenskaya ROM (L-ROM)} with backward Euler time stepping 
is obtained by taking the closure term, $ \uptau^{\mathrm{ROM}}$, in \eqref{eq: Les-ROM} to be
\begin{equation}\label{eq: LZ_closure}
    \uptau^{\mathrm{ROM}} (\wr^{n}) = -(C_S\delta)^\mu \div( \norm{\Grad \wr^{n}}_F^s \Grad \wr^{n}),
\end{equation}
where $ C_S$ is the Smagorinsky constant,  $\delta$ is the ROM lengthscale, and $s$ and $\mu$ are tunable positive exponents. 
Using the above form of $\uptau^{\mathrm{ROM}}$ in \eqref{eq: Les-ROM} and performing integration by parts for the term $\pare{ \uptau^{\mathrm{ROM}} (\wr^{n}), \vr }$, we obtain the following L-ROM: 
For $n = 1, \ldots ,M$,  find 
$\wr^{n} \in \bX_r$ such that
\begin{eqnarray} \label{eq: LZ-ROM}
\begin{split}
   \pare{\frac{\wr^{n} - \wr^{n-1} }{\Delta t }, \vr } + \Rey\inv \pare{\nabla \wr^{n}, \nabla \vr} +
   (C_S \delta)^\mu ( \norm{ \nabla \wr^{n} }_F^s  \nabla \wr^{n}, \nabla \vr)  \\ + b^*(\wr^{n},\wr^{n},\vr)  = (\bff(t^{n}),\vr), \quad \forall \vr \in \bX_r.
\end{split} \quad
\end{eqnarray}

The Ladyzhenskaya model has been studied in the past in the continuous setting \cite{du1991analysis,beirao2005regularity} and for various FOMs \cite{du1990finite,reyesgsm,cao2022continuous}. To our knowledge, this is the first time that this model is examined in the ROM context.

The {\it Smagorinsky ROM (S-ROM)}
is obtained by setting $ \mu=2$ and $s=1$ in the L-ROM \eqref{eq: LZ-ROM}. Namely, the S-ROM can be 
written as follows: 
For $n = 1, \ldots, M$,  find 
$\wr^{n} \in \bX_r$ such that
\begin{eqnarray} \label{eq: S-ROM}
\begin{split}
   \pare{\frac{\wr^{n} - \wr^{n-1} }{\Delta t }, \vr } + \Rey\inv \pare{\nabla \wr^{n}, \nabla \vr} +
   (C_S \delta)^2 ( \norm{ \nabla \wr^{n} }_F  \nabla \wr^{n}, \nabla \vr)  \\ + b^*(\wr^{n},\wr^{n},\vr)  = (\bff(t^{n}),\vr), \quad \forall \vr \in \bX_r.
\end{split} \quad
\end{eqnarray}

We note that 
the S-ROM was first used in a POD setting \cite{noack2002low,ullmann2010pod,wang2012proper}, and then developed for reduced basis  \cite{rebollo2017certified} and variational multiscale methods \cite{ballarin2020certified,iliescu2014variational,rebollo2023certified}. 
As a generalization of the S-ROM, the L-ROM allows more flexibility in choosing the EV magnitude. 
This could alleviate the potential overdiffusivity of the Smagorinsky model, which is well documented in a FOM setting \cite{berselli2006mathematics,john2003large,rebollo2014mathematical,sagaut2006large}.

Next, we prove the verifiability of the 
L-ROM and S-ROM.

\begin{corollary} \label{cor: verifiability}
       The L-ROM  \eqref{eq: LZ-ROM} with initial condition $ \wr^0 = P_r(\bfu_d^0)$ is verifiable: If the time step $\delt$ is chosen such that $ \delt \, d_n \leq 1$ for all $ n= 1, \dots ,M$, then
    \begin{equation}
        \trinorm{\bfe}_{\infty,0}^2 + \Rey\inv \trinorm{\Grad \bfe}_{2,0}^2 \leq K \trinorm{ P_r( \uptau^{\mathrm{FOM}}\pare{\bfu_d} - \uptau^{\mathrm{ROM}}\pare{P_r\pare{\bfu_d} } )}_{2,0}^2,  \label{eq: L-verifiability}
    \end{equation}
    where $ K = \exp\p{\delt \sum_{n=1}^{M}\p{ \frac{d_n}{1-\delt d_n} }}$, $ d_n = C \p{\Rey^{3} \norm{\Grad P_r \p{\bfu_d^{n}} }^4 + 1} $, and $ \bfe = P_r(\bfu_d) - \wr$.
\end{corollary}

\begin{proof}
    The 
    L-ROM uses an eddy viscosity closure model given as
    \begin{equation}
        \uptau^{\mathrm{ROM}} (\wr^{n}) = -(C_S\delta)^\mu \div( \norm{\Grad \wr^{n}}_F^s \Grad \wr^{n}).
    \end{equation}
To prove mean dissipativity of $ \uptau^{\mathrm{ROM}} $, 
we consider $ \ur, \vr \in \bX_r$, and notice that
    \bea
        \pare{ \uptau^{\mathrm{ROM}}(\ur) - \uptau^{\mathrm{ROM}}(\vr), \ur -\vr } = \nonumber \\
        -(C_S\delta)^\mu \pare{  \div\pare{ \norm{\Grad \ur^{n}}_F^s \Grad \ur^{n} - \norm{\Grad \vr^{n}}_F^s \Grad \vr^{n} } , \ur^n -\vr^n }.
    \eea
Using the divergence theorem, we obtain 
    \bea 
         -(C_S\delta)^\mu \pare{  \div\pare{ \norm{\Grad \ur^{n}}_F^s \Grad \ur^{n} - \norm{\Grad \vr^{n}}_F^s \Grad \vr^{n} } , \ur^n -\vr^n } = \nonumber \\
          (C_S\delta)^\mu \pare{  \norm{\Grad \ur^{n}}_F^s \Grad \ur^{n} - \norm{\Grad \vr^{n}}_F^s \Grad \vr^{n}  , \Grad( \ur^n -\vr^n) }.
    \eea    
    By the strong monotonicity Lemma \ref{lemma: Strong Mono}, we have
    \bea 
       (C_S\delta)^\mu \pare{  \norm{\Grad \ur^{n}}_F^s \Grad \ur^{n} - \norm{\Grad \vr^{n}}_F^s \Grad \vr^{n}  , \Grad( \ur^n -\vr^n) } \nonumber \\
        \geq C(C_S\delta)^\mu \norm{\Grad (\ur^n- \vr^n) }_{s+2}^{s+2}\geq 0. 
    \eea
    Thus, the closure term in 
    the L-ROM satisfies the mean dissipativity condition in Definition \ref{def: mean_d}. 
    Therefore, by Theorem \ref{thm:verifiability}, the 
    L-ROM is verifiable.    
\end{proof}

\begin{corollary}    \label{cor: s_verifiability}
       The 
       S-ROM with initial condition $ \wr^0 = P_r(\bfu_d^0)$ is verifiable: 
       If the time step $\delt$ is chosen such that $ \delt \, d_n \leq 1$ for all $ n= 1, \dots ,M$, then
    \begin{equation}
        \trinorm{\bfe}_{\infty,0}^2 + \Rey\inv \trinorm{\Grad \bfe}_{2,0}^2 \leq K \trinorm{P_r(\uptau^{\mathrm{FOM}}\pare{\bfu_d} - \uptau^{\mathrm{ROM}}\pare{P_r\pare{\bfu_d}})}_{2,0}^2,  \label{eq:s_verifiability}
    \end{equation}
    where $ K = \exp\p{\delt \sum_{n=1}^{M}\p{ \frac{d_n}{1-\delt d_n} }}$, $ d_n = C \p{\Rey^{3} \norm{\Grad P_r \p{\bfu_d^{n}} }^4 + 1} $, and $ \bfe = P_r(\bfu_d) - \wr$.
\end{corollary}

\begin{proof}
   The results follows directly from Corollary \ref{cor: verifiability} 
   after choosing $s=1$ and $\mu=2$.
\end{proof}

\begin{remark}
    We note that \eqref{eq: L-verifiability} can be slightly improved if proven directly. If 
    in \eqref{eq: mean-diss} we use Lemma \ref{lemma: Strong Mono} 
    and do not drop the resulting right-hand side,
    we arrive at the following inequality: 
    \bea
        & &\trinorm{\bfe}_{\infty,0}^2 + \Rey\inv \trinorm{\Grad \bfe}_{2,0}^2 + C(C_S\delta)^\mu \delt \sum_{n=1}^{M} \norm{\Grad \bfe^n}_{s+2}^{s+2} \nonumber \\
        &\leq& K \trinorm{P_r(\uptau^{\mathrm{FOM}}\pare{\bfu_d} - \uptau^{\mathrm{ROM}}\pare{P_r\pare{\bfu_d} }}_{2,0}^2.  \label{eq:verifiability_ex}
    \eea
\end{remark}

\section{Limit Consistency of LES-ROMs}
\label{sec: Limit_Con}

The limit consistency of closure models has a long history in classical CFD (see, e.g., \cite{Kaya2002verifiability} and the survey in \cite{berselli2006mathematics}). 
In this section, we introduce the concept of \emph{limit consistency} for the first time in the ROM context.
Specifically, we first define the limit consistency for LES-ROMs, and then prove that, under suitable assumptions, S-ROM and L-ROM are limit consistent.

In classical LES, 
the concept of limit consistency was introduced in \cite{Kaya2002verifiability} to identify additional conditions on the closure model that ensure \( \bfw \to \bfu \) as \( \delta \to 0 \). The term 
limit consistency was chosen because this property is a standard result in conventional filtering, i.e., \( \overline{\bfu} \to \bfu \) as \( \delta \to 0 \). Thus, if \( \bfw \) is a consistent approximation in the limit, it should exhibit the same behavior.

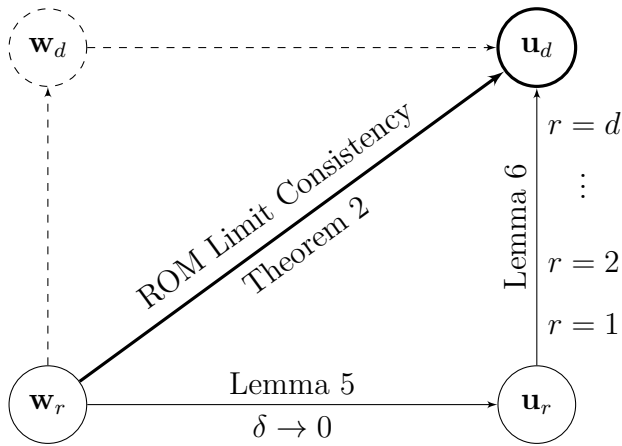
\begin{figure}[h!]
    \centering
    \begin{tikzpicture}
    \node[vertex] (1) at (0,0) {$\wr$};
    \node[vertex] (2) at (6.5,0) {$\ur$};
    \node[vertex,dashed] (3) at (0,4.75) {$\bfw_d$};
    \node[vertex,very thick] (4) at (6.5, 4.75) {$\bfu_d$};
    
    \draw[edge] (1) -- (2) node[midway, below] {$\delta \to 0 $} node[midway, above] {Lemma 5};
    \draw[edge,dashed] (1) -- (3);
    \draw[edge,very thick] (1) -- (4) node[pos=.5, above, sloped, rotate=0] {ROM Limit Consistency}
    node[pos=.5, below, sloped, rotate=0] {Theorem 2};
    \draw[edge] (2) -- (4)  node[pos=.5, above, rotate=90] {Lemma 6} 
    node[pos=.15, right, rotate=0] {$r= 1 $}
    node[pos=.375, right, rotate=0] {$r= 2$}
    node[pos=.675, right, rotate=0] {$ \quad \vdots $}
    node[pos=.875, right, rotate=0] {$r= d$};
    \draw[edge,dashed] (3) -- (4);
    \end{tikzpicture}
    \caption{ Schematic of limit consistency in the LES--ROM context.}
    \label{fig: LC}
\end{figure}

Extending 
the limit consistency concept to LES-ROMs is a nontrivial task 
because this introduces a new dimension needed to account for resolution of the discretization, as can be seen in Fig.~\ref{fig: LC}. Here, the horizontal component of the diagram represents a natural extension of limit consistency as the limit of 
ROM lengthscale $ \delta \geq 0$ can be studied for a fixed dimension $ 1 \leq r \leq d$. 
In Lemma \ref{lem: wr2ur}, we prove  
the convergence rate of $ \wr \to \ur$ as $ \delta \to 0$.
The vertical component of Fig.~\ref{fig: LC} represents the new bounds proven in Lemma \ref{lem: ur2ud} to account for the resolution of ROM
discretization.
Specifically, in Lemma \ref{lem: ur2ud} we prove that by increasing 
the reduced dimension $ r=
1,2,\dots,d
$, 
the solution of the $r$-dimensional G-ROM gets closer to the solution of the $d$-dimensional G-ROM. 
We note that this convergence appears naturally because, in a POD context, the convergence of the G-ROM solution to the FOM data is well documented~\cite{KV01,volkwein2013proper}.
Combining 
Lemmas \ref{lem: wr2ur} and \ref{lem: ur2ud}, we can 
then prove Theorem \ref{thm:LimitConsistency}, which yields 
the first  
bounds of a 
ROM limit consistency representing the diagonal path in Fig.~\ref{fig: LC}.

Next, we propose a definition of limit consistency for the LES-ROMs developed in Section~\ref{ssec: les_rom}.

\begin{definition}[LES-ROM Limit Consistency]
    The LES-ROM~\eqref{eq: Les-ROM} is limit consistent if 
    there exist $m$ and $k$ such that the LES-ROM 
    solution, $\bw_r$, satisfies an error bound of the form
    \begin{eqnarray}
        \trinorm{ \bu_d - \bw_r }_{m,k}
        \leq B_{LC}, 
        \label{eqn:definition-limit-consistency-1}
    \end{eqnarray}
    where $\bu_d$ is the ``true solution" (i.e., the solution of the $d$-dimensional G-ROM~\eqref{eq: gromu}), and the limit consistency bound, $B_{LC}$, which depends on the lengthscale, $\delta$, and the ROM dimension, $r$, satisfies the following two conditions:
    \begin{enumerate}
        \item[(1)] $B_{LC}$ 
        decreases as $\delta \longrightarrow 0$ and $r$ increases (i.e., $r=1, r=2, \ldots, r=d$).
        \item[(2)] When $r=d$, $B_{LC} \longrightarrow 0$ 
        when $\delta \longrightarrow 0$.  
    \end{enumerate}
    \label{definition:limit-consistency}
\end{definition}

\begin{remark}
    We emphasize that there are two differences between the LES-ROM concept of limit consistency %
    in Definition~\ref{definition:limit-consistency} (and illustrated in Fig.~\ref{fig: LC}) and the concept of classical LES limit consistency.
    The first difference is that in the classical LES framework limit consistency is studied with respect to only one parameter: the lengthscale, $\delta$.
    In contrast, in the LES-ROM framework limit consistency is studied with respect to two parameters: 
    The limit consistency bound, $B_{LC}$, in \eqref{eqn:definition-limit-consistency-1} depends on 
    the lengthscale, $\delta$, and the ROM dimension, $r$.
    The second difference is that, in LES-ROM limit consistency, the concept of ``limit'' is understood in a broad sense.
    Indeed, because the ROM dimension, $r$, belongs to a finite discrete set $\curly{ 1, \ldots, d}$, to study the ``limit'' behavior of $r$, we consider increasing $r$ values, i.e., $r=1, r=2, \ldots, r=d$ (see the continuous vertical line in Fig.~\ref{fig: LC}).
    \label{remark:limit-consistency}
\end{remark}

We also note that, as 
this is the first time the concept of limit consistency is investigated in the ROM context, there remain many open questions on the topic. A natural question would be on coupling  
the ROM lengthscale $ \delta$ to the reduced dimension $r$, 
as this would reduce the study of LES-ROM limit consistency to increasing the ROM dimension without requiring that $\delta \longrightarrow 0$ at the same time.
One possible strategy along these lines would be using the recently introduced energy-based ROM lengthscale~\cite{mou2023energy}.
The authors also posit an alternative direction for the proof of Theorem \ref{thm:LimitConsistency} as can be seen by the dotted lines in Fig.~\ref{fig: LC}: 
First, the LES-ROM is resolved, 
i.e., we prove that as $r$ increases, the $r$-dimensional LES-ROM solution, $\bw_r$, approaches the $d$-dimensional LES-ROM solution, $\bw_d$  
(the vertical dotted line in Fig.~\ref{fig: LC}). 
Then, Lemma \ref{lem: wr2ur} is used for $ r=d$ to prove that, as $\delta \longrightarrow 0$, $\bw_d \longrightarrow \bu_d$ (the horizontal dotted line in Fig.~\ref{fig: LC}). 

In Lemma~\ref{lem: wr2ur}, for a fixed ROM dimension, we prove the convergence of the LES-ROM solution to the G-ROM solution as the ROM lengthscale in the LES-ROM goes to zero.
We note that this corresponds to the horizontal continuous line in Fig.~\ref{fig: LC}.
To prove Lemma~\ref{lem: wr2ur}, we first make the following assumption, which
can be thought of as the discrete version of Assumption 3 in \cite{Kaya2002verifiability}.

\begin{assumption}\label{as: assumption 1}
Suppose that 
the  solution 
to the G-ROM
\eqref{eq: gromu-r}, 
$ \ur^n \in \bX_r$, 
satisifes:

$$ \delt \sum_{n=1}^{M} \norm{\uptau^{\mathrm{ROM}}(\ur^n)}^2 \to 0 \text{ as } \delta \to 0.$$
\end{assumption}

\begin{lemma}\label{lem: wr2ur}
For a fixed 
ROM dimension, r, suppose that $ \wr^n \in \bX_r$ is the solution to the LES-ROM \eqref{eq: Les-ROM}, and $ \ur^n \in \bX_r$ the solution to the G-ROM \eqref{eq: gromu-r}. If the closure model, $\uptau^{\mathrm{ROM}}$, satisfies the mean dissipativity condition \ref{def: mean_d} and Assumption \ref{as: assumption 1}, then 

 \bea 
 \trinorm{\ur - \wr}_{\infty,0}^2  
  \leq  K  \trinorm{\uptau^{\mathrm{ROM}}(\ur)}_{2,0}^2.  
\eea

In particular, $\trinorm{\ur - \wr}_{\infty,0}^2 \rightarrow 0$ as $ \delta \to 0$.

\end{lemma}

\begin{proof}

Subtracting the LES-ROM \eqref{eq: Les-ROM} from the G-ROM \eqref{eq: gromu-r} yields 
\bea
\frac{1}{\delt} \pare{ \ur^n -\wr^n  - (\ur^{n-1} - \wr^{n-1}) ,\vr} + \Rey\inv \pare{\Grad (\ur^n - \wr^n), \Grad\vr} + b^*(\ur^n,\ur^n,\vr) \nonumber\\ 
- b^{*}(\wr^n,\wr^n,\vr) + \pare{\uptau^{\mathrm{ROM}}(\ur^n),\vr} - \pare{\uptau^{\mathrm{ROM}}(\wr^n),\vr} = \pare{\uptau^{\mathrm{ROM}}(\ur^n),\vr}. \qquad
\eea
Letting $ \ur^n - \wr^n = \bfe^n$ and 
choosing $ \vr = \bfe^n$ results in 
\bea
\frac{1}{\delt} \pare{ \bfe^n  - \bfe^{n-1} ,\bfe^n} + \Rey\inv \norm{\Grad \bfe^n}^2 + b^*(\ur^n,\ur^n,\bfe^n) - b^{*}(\wr^n,\wr^n,\bfe^n) \nonumber\\ 
 + \pare{\uptau^{\mathrm{ROM}}(\ur^n) - \uptau^{\mathrm{ROM}}(\wr^n),\bfe^n} = \pare{\uptau^{\mathrm{ROM}}(\ur^n),\bfe^n} .
\eea
Rewriting the trilinear terms as before, and 
applying the Cauchy-Schwarz and Young's inequalities to the time derivative term gives the following: 
\bea
\frac{1}{2\delt} \pare{ \norm{\bfe^n}^2 - \norm{\bfe^{n-1}}^2} + \Rey\inv \norm{\Grad \bfe^n}^2 + b^*(\bfe^n,\ur^n,\bfe^n) \nonumber\\ 
 + \pare{\uptau^{\mathrm{ROM}}(\ur^n) - \uptau^{\mathrm{ROM}}(\wr^n),\ur^n - \wr^n} \leq\pare{\uptau^{\mathrm{ROM}}(\ur^n),\bfe^n} .
\eea

Rewriting the equation, using the mean dissipativity \eqref{eq: mean_dis}, taking the absolute value, and using the triangle inequality gives 

\bea
\frac{\norm{\bfe^n}^2 - \norm{\bfe^{n-1}}^2}{2\delt} + \Rey\inv \norm{\Grad \bfe^n}^2  
  \leq \abs{ b^*(\bfe^n,\ur^n,\bfe^n) } + \abs{\pare{\uptau^{\mathrm{ROM}}(\ur^n),\bfe^n}}. \quad
\eea

We bound the nonlinear term on the right-hand side in the same fashion as before, and using Cauchy-Schwarz and Young's inequalities on the second term results in
\bea
\frac{1}{2\delt} \pare{ \norm{\bfe^n}^2 - \norm{\bfe^{n-1}}^2} + \frac{1}{2}\Rey\inv \norm{ \Grad \bfe^n}^2  
  &\leq& \pare{\frac{1}{2} + C \Rey^3 \norm{\Grad \ur^n}^4 }\norm{\bfe^n}^2 \nonumber \\
  &+& \frac{1}{2}\norm{\uptau^{\mathrm{ROM}}(\ur^n)}^2.
\eea
 Multiplying by $ 2\delt$ and summing from $ 1$ to $ M$ gives 
 \bea
 \norm{\bfe^M}^2 + \Rey\inv \delt \sum_{n=1}^{M} \norm{\Grad \bfe^n}^2  
  \leq \delt \sum_{n=1}^{M} \pare{ 1 + C \Rey^3 \norm{\Grad \ur^n}^4 }\norm{\bfe^n}^2 \nonumber \\
  + \delt \sum_{n=1}^{M} \norm{\uptau^{\mathrm{ROM}}(\ur^n)}^2,
\eea
because $ \norm{\bfe^0} = 0$.
Assuming a $ \delt$ sufficiently small, $ d_n \delt \leq 1 $, we apply the  discrete Gronwall's inequality, giving
 \bea \label{eq: LC_last}
 \norm{\bfe^M}^2 + \Rey\inv \delt \sum_{n=1}^{M} \norm{ \Grad \bfe^n}^2  
  \leq  K \delt \sum_{n=1}^{M} \norm{\uptau^{\mathrm{ROM}}(\ur^n)}^2, 
\eea
where $ K = \exp\p{\delt \sum_{n=1}^{M}\p{ \frac{d_n}{1-\delt d_n} }}$ and $ d_n = 1 +  C \Rey^{3} \norm{\Grad \bfu_d^{n} }^4 $.
Taking $ \delta \to 0$ completes the proof thanks to Assumption 1.
\end{proof}

In Lemma \ref{lem: ur2ud}, we 
quantify the error between G-ROMs of different dimensions.
We note that this corresponds to the vertical continuous line in Fig.~\ref{fig: LC}.

{

\begin{lemma}\label{lem: ur2ud}

Consider $ \ud^n \in \bX_d$, the solution to the d-dimensional G-ROM \eqref{eq: gromu}, and $ \ur^n \in \bX_r$, the solution to the r-dimensional G-ROM \eqref{eq: gromu-r}, for $ r = 1,\dots,d$, and $ n = 1,\dots,M$.
Then, for a sufficiently small time step, we have  
\bea
  \trinorm{\bfe}_{\infty,0}^2 
  \leq \trinorm{\bfeta}_{\infty,0}^2 + K C  \bigg( \Rey\inv \trinorm{\Grad \bfeta}_{2,0}^2  +  \Rey^3 \trinorm{\Grad \bfeta}_{\infty,0}^2  \bigg) ,
\eea
where $\bfe^n = \ud^n - \ur^n $, $\bfeta^n = \ud^n - P_r(\ud^n) $, $ K = \exp\p{\delt \sum_{n=1}^{M}\p{ \frac{d_n}{1-\delt d_n} }}$, and $ d_n = C \Rey^{3} \norm{\Grad  \bfu_d^{n} }^4 $. 
\end{lemma}

\begin{proof}
We subtract the $r$-dimensional G-ROM \eqref{eq: gromu-r} from the $d$-dimensional G-ROM \eqref{eq: gromu}: For all $\bv_r \in \bX_r$,
\bea
\hspace*{-0.6cm}
\pare{ \frac{\bfe^n  - \bfe^{n-1}}{\delt} ,\vr} + \Rey\inv \pare{\Grad \bfe^n, \Grad\vr} 
 + b^*(\ud^n,\ud^n,\vr) - b^{*}(\ur^n,\ur^n,\vr) = 0 .
\eea
Decomposing the error as $ \bfe^n = \ud^n - P_r(\ud^n) - (\ur^n - P_r(\ud^n)) := \bfeta^n - \phir^n $, and choosing $ \vr = \phir^n$ results in 
\bea
\frac{1}{\delt} \pare{ \phir^n  - \phir^{n-1} ,\phir^n} + \Rey\inv \norm{\Grad \phir^n}^2 = \frac{1}{\delt} \pare{ \bfeta^n - \bfeta^{n-1}, \phir^n } \nonumber \\
+ \Rey\inv \pare{\Grad \bfeta^n, \Grad \phir^n} + b^*(\ud^n,\ud^n,\phir^n) - b^{*}(\ur^n,\ur^n,\phir^n). \label{eq: udur_error1}
\eea
We note that $ \frac{1}{\delt} \pare{ \bfeta^n - \bfeta^{n-1}, \phir^n } = 0 $ from Definition \ref{def: ROM_Proj}. 
We rewrite the trilinear terms as before, and 
decompose the error as follows: 
\bea
& &b^{\ast}( \ud^{n}, \ud^{n},\phir^n) - b^{\ast}(   \ur^{n},  \ur^{n},\phir^n)
= b^{\ast}( \bfe^{n}, \ud^{n},\phir) + b^{\ast}(   \ur^{n},  \bfe^{n},\phir^n) = \nonumber \\
& &b^{\ast}( \bfeta^{n}, \ud^{n},\phir^n) 
- b^{\ast}( \phir^{n}, \ud^{n},\phir^n)
+ b^{\ast}( \ur^{n}, \bfeta^{n},\phir^n) 
- b^{\ast}(  \ur^{n},  \phir^{n},\phir^n). \label{eq: trilinear}
\eea 
Noting that $ b^*(\ur^{n},  \phir^{n},\phir^n) =0$,
substituting \eqref{eq: trilinear} into \eqref{eq: udur_error1}, 
and applying the Cauchy-Schwarz and Young's inequalities to the time derivative term leads to
\bea
\frac{1}{2\delt} \pare{ \norm{\phir^n}^2 - \norm{\phir^{n-1}}^2} + \Rey\inv \norm{\Grad \phir^n}^2 \leq  \Rey\inv \pare{\Grad \bfeta^n, \Grad \phir^n} \nonumber \\
+ b^{\ast}( \bfeta^{n}, \ud^{n},\phir^n) 
- b^{\ast}( \phir^{n}, \ud^{n},\phir^n)
+ b^{\ast}( \ur^{n}, \bfeta^{n},\phir^n).
\eea
Rewriting the equation, taking the absolute value, and using the triangle inequality gives
\bea
& &\frac{1}{2\delt} \pare{ \norm{\phir^n}^2 - \norm{\phir^{n-1}}^2} + \Rey\inv \norm{\Grad \phir^n}^2 \leq  \Rey\inv \abs{\pare{\Grad \bfeta^n, \Grad \phir^n} }\nonumber \\
& &+ |b^{\ast}( \bfeta^{n}, \ud^{n},\phir^n)| 
+ |b^{\ast}( \phir^{n}, \ud^{n},\phir^n)|
+ \abs{b^{\ast}( \ur^{n}, \bfeta^{n},\phir^n)} \nonumber \\
&=& |T_1|+|T_2|+|T_3|+|T_4|. \label{eq: udur_error2}
\eea
We then bound the four terms on the right-hand side using 
Lemma~\ref{TRIL} for the trilinear terms, as well as the Cauchy-Schwarz and Young's inequalities throughout:
\bea
|T_1| & \leq & \frac{\Rey\inv}{8}||\nabla   \phir^{n}||^2+C\Rey\inv||\nabla \bfeta^{n}||^2, \label{eq: t1_bound}\\
   |T_2| &\leq& C ||\nabla \bfeta^{n}|| ||\nabla \ud^{n}|| ||\nabla    \phir^{n}|| \nonumber \\
   &\leq &  \frac{\Rey\inv}{8}||\nabla   \phir^{n}||^2 + C\Rey ||\nabla \bfeta^{n}||^2 ||\nabla \ud^{n}||^2, \label{eq: t2_bound} \\
  |T_3| &\leq&  C \sqrt{||  \phir^{n}|| ||\nabla   \phir^{n}||}||\nabla \ud^{n}|| ||\nabla    \phir^{n}||  \nonumber \\
  & \leq &  \frac{
  \Rey\inv}{8}||\nabla   \phir^{n}||^2 + C \Rey^{3}||\nabla \ud^{n}||^4 ||  \phir^{n}||^2, \label{eq: t3_bound} \\
  |T_4| &\leq& C  ||\nabla   \ur^{n}||||\nabla \bfeta^{n}|| ||\nabla    \phir^{n}|| \nonumber \\
  &\leq &  \frac{\Rey\inv}{8}||\nabla   \phir^{n}||^2 + C\Rey ||\nabla   \ur^{n}||^2 ||\nabla \bfeta^{n}||^2. \label{eq: t4_bound}
\end{eqnarray}

Substituting the bounds \eqref{eq: t1_bound} - \eqref{eq: t4_bound} into \eqref{eq: udur_error2}, multiplying by $ 2\delt$, 
summing from $n=  1$ to $ M$, and recalling that $ \norm{\phir^0} =0 $, gives 
 \bea
 & & \norm{\phir^M}^2 + \Rey\inv \delt \sum_{n=1}^{M} \norm{\Grad \phir^n}^2  
  \leq C \Rey\inv \delt \sum_{n=1}^{M} \norm{\Grad \bfeta^n}^2 \nonumber\\
  & & + C \Rey \delt \sum_{n=0}^{M} \norm{\Grad \ud^n}^2 \norm{\Grad \bfeta^n}^2 +
  \delt \sum_{n=1}^{M} \pare{ C \Rey^3 \norm{\Grad \bfu_d^n}^4 }\norm{\phir^n}^2\nonumber \\
  & &+ C\Rey \delt \sum_{n=1}^{M}  ||\nabla   \ur^{n}||^2 ||\nabla \bfeta^{n}||^2.
\eea
We now use the stability bounds from Lemma \ref{lemma: Stability}, yielding 
 \bea
 & & \norm{\phir^M}^2 + \Rey\inv \delt \sum_{n=1}^{M} \norm{\Grad \phir^n}^2  
  \leq C \Rey\inv \delt \sum_{n=1}^{M} \norm{\Grad \bfeta^n}^2 \nonumber\\
  & & + C \Rey^3 \trinorm{\Grad \bfeta}_{\infty,0}^2 +
  \delt \sum_{n=1}^{M} \pare{ C \Rey^3 \norm{\Grad \bfu_d^n}^4 }\norm{\phir^n}^2.
\eea

Assuming a sufficiently small time step ($ d_n \delt :=  C\Rey^3 \norm{\Grad \ud^n }^4  \delt < 1 $), we apply the  discrete Gronwall's inequality in Lemma \ref{lemma: discreteGronwall}, giving  
 \bea \label{eq: ur-ud_last}
 \norm{\phir^M}^2 + \Rey\inv \delt \sum_{n=1}^{M} \norm{ \Grad \phir^n}^2  
  \leq  K C  \bigg(   \Rey\inv \trinorm{\Grad \bfeta}_{2,0}^2  + \Rey^3  \trinorm{\Grad \bfeta}_{\infty,0}^2  \bigg), \ \ 
\eea
where $ K = \exp\p{\delt \sum_{n=1}^{M}\p{ \frac{d_n}{1-\delt d_n} }}$ and $ d_n = C \Rey^{3} \norm{\Grad \bfu_d^{n} }^4 $. 
Applying the triangle inequality gives the result.
\end{proof}
}

Combining Lemma \ref{lem: wr2ur} and Lemma \ref{lem: ur2ud} (i.e., the continuous horizontal and vertical lines in Fig.~\ref{fig: LC}, respectively), we can prove Theorem~\ref{thm:LimitConsistency} (i.e., the diagonal in Fig.~\ref{fig: LC}):

\begin{theorem}
\label{thm:LimitConsistency}
     For $ n = 1, \dots M$, let $ \ud^n \in \bX_d $ and $ \wr^n \in \bX_r $ be the solutions to \eqref{eq: gromu} and \eqref{eq: Les-ROM}, respectively. Then, for a sufficiently small time step, 
     and $ r = 1, \dots,d$, 
    if the closure model, $\uptau^{\mathrm{ROM}}$, satisfies the mean dissipativity condition in Definition \ref{def: mean_d} and Assumption \ref{as: assumption 1}, we have the following error bound:
    \bea\label{eq: ROM_LC}
    \trinorm{\bfu_d - \wr}_{\infty,0}^2 \leq CK \bigg( \trinorm{\bfeta}_{\infty,0}^2 +  \Rey\inv \trinorm{\Grad \bfeta}_{2,0}^2 \nonumber \\
    +  \Rey^3 \trinorm{\Grad \bfeta}_{\infty,0}^2
    +  \trinorm{\uptau^{\mathrm{ROM}}(\ur)}_{2,0}^2 \bigg),
    \eea
    where $\bfeta^n = \ud^n - P_r(\ud^n) $, and $ K$ depends exponentially on $ \Rey^3$. 
\end{theorem}

\begin{proof}
    Considering the difference 
    between the solution 
    to the $d$-dimensional G-ROM, $\ud^n$, and the solution 
    to the $r$-dimensional LES-ROM, $\wr^n$, 
    adding and subtracting 
    the solution 
    to the $r$-dimensional G-ROM, $\ur$,
    and applying the triangle inequality gives  
    \beas
    \trinorm{\bfu_d - \wr}_{\infty,0}^2 &\leq& \trinorm{\ur - \wr}_{\infty,0}^2 + \trinorm{\bfu_d - \ur}_{\infty,0}^2.
    \eeas
    Applying Lemmas \ref{lem: wr2ur} and \ref{lem: ur2ud} for the first and second term on the right-hand side of the above inequality, respectively, 
    yields the result. 
\end{proof}

\begin{remark}
    The error bound \eqref{eq: ROM_LC} is an {\it a priori} bound:
    The right-hand side of the error bound \eqref{eq: ROM_LC} depends on the 
    fixed ROM parameters (e.g., $\Delta t, r, d, \bu_d, \bu_r, \delta$), but does not depend on the LES-ROM solution, $\bw_r$.
    \label{remark:a-priori}
\end{remark}

\begin{theorem}
\label{thm:Thm3}
Assume that $ \Grad \ur^n \in L^{2s+2}(0,T; L^{2s+2}(\ohm))$. 
For $ \mu >0$ and 
$s \geq 0 $,  
the L-ROM and S-ROM are
limit consistent.
\end{theorem}

\begin{proof}
    Because $ \Grad \ur^n \in L^{2s+2}(0,T; L^{2s+2}(\ohm))$, we have
    \bea
        \delt \sum_{n=1}^{M} \norm{\uptau^{\mathrm{ROM}}(\ur^n)}^2 &=& \delt \sum_{n=1}^{M} (C_S\delta)^{2\mu} \int_{\Omega} \norm{\Grad \ur^n}_F^{2s} \abs{\Grad \ur^n}^2 d\bx \nonumber\\
        &=&(C_S\delta)^{2\mu} \delt \sum_{n=1}^{M} \norm{\Grad \ur^n}_{2s+2}^{2s+2} \nonumber\\
        &=&(C_S\delta)^{2\mu} \trinorm{\Grad \ur^n}_{2s+2,2s+2}^{2s+2}. \label{eq: L-ROM_LC}
    \eea

Thus, because 
$(C_S\delta)^{2\mu} \trinorm{\Grad \ur}_{2s+2,2s+2}^{2s+2} \to 0 $  as $ \delta \to 0$, 
the closure model \eqref{eq: LZ_closure} in the L-ROM \eqref{eq: LZ-ROM} satisfies Assumption \ref{as: assumption 1}.
Because the L-ROM is mean dissipative by Corollary~\ref{cor: verifiability}, 
assuming sufficient regularity of $ \ur$ and using 
Theorem~\ref{thm:LimitConsistency} implies that the L-ROM is limit consistent.
To prove limit consistency of the S-ROM, we 
pick $ \mu =2 $ and $ s=1$ in the L-ROM \eqref{eq: LZ-ROM}.
\end{proof}

\begin{remark}
 By Lemma \ref{lem: wr2ur}, \eqref{eq: L-ROM_LC}  gives a convergence rate for $\norm{\wr - \ur}$ as $ O(\delta^\mu)$ for the L-ROM, and $ O(\delta^2)$ for the S-ROM.
   \label{remark:delta-convergence}
\end{remark}

\section{Numerical Results} \label{sec-numres}

In this section, we investigate whether
verifiability (Theorem \ref{thm:verifiability}) and limit consistency (Theorem \ref{thm:LimitConsistency}) are recovered numerically for the S-ROM and L-ROM. 

We investigate the verifiability and the limit consistency of the S-ROM and L-ROM for two test problems: the one-dimensional (1D) Burgers
equation at 
$\nu= 2 \times10^{-3}$ (Section~\ref{subsec:1dburg}), and the two-dimensional (2D) lid-driven cavity
at Reynolds number $Re = 15,000$ (Section~\ref{subsec:2dldc}). 
The Burgers equation is a nonlinear advection-diffusion
problem that we will test in a low viscosity setting, making it a challenging problem for the standard G-ROM to handle in the under-resolved regime. 
The lid-driven cavity problem displays chaotic phenomena at high Reynolds numbers, which can give rise to spurious oscillations and inaccurate solutions in the under-resolved regime, and can yield similar challenges for the G-ROM.

\subsection{
Evaluation Criteria}
For verifiability,  
we use the following metrics. To quantify the error of LES-ROMs, that is, the first term on the left-hand side of inequality \eqref{eq:verifiability}, we use the following average $L^2$ norm: 
\begin{align}
    \ROMError := \frac{1}{M+1}\sum^M_{n=0}\| {P}_{r}(\bu^n_d)-\bw^n_r\|^2,
    \label{eq:error_l2}
\end{align}
where $\bu^n_d$ is the $d$-dimensional G-ROM velocity field, 
$\bw^n_r$ is the $r$-dimensional LES-ROM 
velocity field, {and} 
$P_r$ is the ROM $L^2 $ projection
(Definition \ref{def: ROM_Proj})
onto the $r$-dimensional reduced space. 
We remark that $\ROMError$ is measuring the error between the LES-ROM solution, $\bw_r$, and the large-scale component of the 
true solution, $P_r(\bu_d)$. 

To quantify the closure error of the S-ROM and L-ROM, that is, the term on the right-hand side of inequality (\ref{eq:verifiability}), we use the following metric:
\begin{align}
    \ClosureError := \frac{1}{M+1}\sum^M_{n=0}\|
    P_r\left(\uptau^{\mathrm{FOM}}(\bu_d^{n})\right) -P_r\left(\uptau^{\mathrm{ROM}}(P_r\bu^n_d)\right) 
    \|^2,
    \label{eq:closure_l2}
\end{align}
where $P_r(\uptau^{\mathrm{FOM}})$ and $P_r(\uptau^{\mathrm{ROM}})$ are the projection of $\uptau^{\mathrm{FOM}}$ and $\uptau^{\mathrm{ROM}}$ onto the $r$-dimensional reduced space, respectively, and 
are computed through the following identities: 
\begin{align}
    \left({P}_r\left(\uptau^{\mathrm{FOM}}(\bfu_d^{n})\right),\bphi_i \right) & = \left(P_r( \bfu_d^n \cdot \Grad \bfu_d^n  ) - P_r(\bfu_d^n) \cdot \Grad P_r(\bfu_d^n),\bphi_i \right), \\ 
    \left({P}_r\left(\uptau^{\mathrm{ROM}}(P_r\bfu^n_d)\right), \bphi_i \right) & =\left((C_S\delta)^\mu \norm{\Grad P_r\bfu^{n}_d}^s_F \Grad P_r\bfu^{n}_d,\nabla \bphi_i\right), 
\end{align}
for all $i=1,\ldots,r$. We note that the ROM error $\ROMError$ and the closure error $\ClosureError$ depend on the reduced space dimension, $r$. {To simplify the notation,} we suppress this dependency.

To illustrate numerically that both the S-ROM and L-ROM are verifiable as proven in Theorem~\ref{thm:verifiability}, we need to show that as $\ClosureError$ (\ref{eq:closure_l2}) decreases, so does $\ROMError$ (\ref{eq:error_l2}). Specifically, according to (\ref{eq:verifiability}), we should see $\log(\ClosureError)$ and $\log(\ROMError)$ obey the following relation:
\begin{equation}
    \log(\ROMError) \le \alpha \log(\ClosureError) + \beta, \label{eq:relation}
\end{equation}
with $\alpha \approx 1$ and some $\beta >0$. For a fixed {lengthscale} $\delta$, we aim to show that (\ref{eq:relation}) holds with $\alpha \approx 1$ as $r$ varies.

For limit consistency, 
we use the following metrics. To quantify the error between the $r$-dimensional LES-ROMs and the $ d$-dimensional G-ROM, that is, the left-hand side of inequality \eqref{eq: ROM_LC}, we use the following average $L^2$ norm: 

\begin{align}
    \eplimit := \frac{1}{M+1}\sum^M_{n=0}\|\bu^n_d-\bw^n_r\|^2.
    \label{eq:limit_left-hand side}
\end{align}
To quantify the term on the right-hand side of inequality \eqref{eq: ROM_LC}, we use the following metric: 

\begin{align}
     \etalimit & := \trinorm{\bfeta}_{\infty,0}^2 +  \Rey\inv \trinorm{\Grad \bfeta}_{2,0}^2 
    \nonumber \\ & +  
    \Rey^3 \trinorm{\Grad \bfeta}_{\infty,0}^2
    +  \trinorm{\uptau^{\mathrm{ROM}}(\ur)}_{2,0}^2.
     \label{eq:limit_right-hand side}
\end{align} 
To illustrate numerically that S-ROM and L-ROM are limit consistent as we have defined in Section \ref{sec: Limit_Con} and proven in Theorem~\ref{thm:LimitConsistency}, we 
{will} show 
that, given pairs of $(r, \delta)$ points with $r$ approaching $d$ and $\delta \to 0$, $\eplimit$ is bounded by the a priori bound $\etalimit$, and $\eplimit$ trends towards 0.
We note that, because the 
$Re^3$ factor on the right-hand side of \eqref{eq:limit_right-hand side} is relatively large, in our numerical investigation we suppress the $Re$ dependency when we plot $\etalimit$ to ensure a clearer display of the results. 

\begin{remark}
    We note that the metrics defined in (\ref{eq:error_l2})--(\ref{eq:closure_l2}) and (\ref{eq:limit_left-hand side}) depend on the 
    $d$-dimensional G-ROM solution, $\bu^n_d$. 
    {For computational efficiency, in our numerical investigation} 
    we use the projection of $\bu^n_{\mathrm{FOM}}$ onto the $d$-dimensional reduced space, $P_d(\bu^n_{\mathrm{FOM}})$, for $\bfu^n_d$ to avoid running the $d$-dimensional G-ROM.
\end{remark}

\subsection{S-ROM/L-ROM
Computational Implementation}
\label{sec:implementation}

The L-ROM \eqref{eq: LZ-ROM} {is a discretized version of} the following semi-discrete
system:
\begin{align}
  \de{\uwr}{t} & + \left(\frac{1}{\rm Re}A + \wtA(\bw_r)\right) \uwr  + \left(\frac{1}{\rm Re} \ua_0 + \wtua_0(\bw_r)\right)  
  + (\uwr)^T B\uwr \nonumber \\ & + B_1\uwr + B_2\uwr + \ub_0
  = \uzero,
  \label{eq:NSE_ODE_semi-discretized}
\end{align}

where the reduced operators are defined as
\begin{align}
    A_{ij} & := \left(\nabla \bphi_i, \nabla \bphi_j \right),\quad \wtA_{ij}(\bw_r)  := (C_S \delta)^\mu \left(\nabla \bphi_i, \|\nabla \bw_r \|^s_F 
    \nabla \bphi_j
    \right), \\
    a_{0,i} &:= \left(\nabla \bphi_i, \nabla \bphi_0 \right), \quad \widetilde{a}_{0,i}(\bw_r) := 
    (C_S \delta)^\mu \left(\nabla \bphi_i, \|\nabla \bw_r \|^s_F \nabla \bphi_0 \right), \\ B_{ikj} & := \left(\bphi_i, (\bphi_k \cdot \nabla) \bphi_j\right), \quad B_{1,ij} := \left(\bphi_i, (\bphi_0 \cdot \nabla) \bphi_j\right) \\ 
    B_{2,ik} & := \left(\bphi_i, (\bphi_k \cdot \nabla) \bphi_0\right) \quad  b_{0,i} := \left(\bphi_i, (\bphi_0 \cdot \nabla) \bphi_0 \right),
\end{align}
for all $i,j,k=1,\ldots,r$.
We note that the reduced mass matrix is simply 
{the} identity matrix due to the orthonormality of the POD basis functions.
In addition, the Smagorinsky constant $C_S$ is set to $1$.
We parameterize the S-ROM with $\mu=2$ and $s=1$, and the L-ROM with $\mu = 10/3$ and $s=2$. The parameter choice for L-ROM is drawn from L-FOM results in  \cite{layton2008introduction,reyesgsm}.

\begin{remark}
    We note that the computational cost of solving the S-ROM and L-ROM scales with the {number of} FOM degrees of freedom because the eddy viscosity depends on $\|\nabla \bu_r\|^s_F$. In this study, we do not consider hyper-reduction techniques because the purpose is to demonstrate verifiability and limit consistency.
    {We emphasize, however, that to ensure the efficiency of the S-ROM and L-ROM in realistic applications (e.g., turbulent flow simulations), both ROMs should be equipped with hyper-reduction \cite{chaturantabut2012state}.}
\end{remark}
In the numerical investigation, we use the backward Euler time discretization in the 1D Burgers equation, and a semi-implicit scheme, the BDF3/EXT3 
scheme, in the 2D lid-driven cavity.  
Although the verifiability and the limit consistency of the S-ROM and L-ROM
proven in Theorem~\ref{thm:verifiability} and Theorem~\ref{thm:LimitConsistency}
were using the backward Euler time discretization, we believe that the
mathematical arguments used in the proofs can be extended to higher-order time
discretizations such as the one considered here.

\subsection{1D Burgers Equation}
\label{subsec:1dburg}

In this section, we consider the following 1D viscous Burgers equation: 

\begin{equation}
    u_t - \nu \, u_{xx} + u \, u_x = 0, \quad 
    x \in (0,1), \ t > 0, 
\end{equation}

where $u$ is the unknown and $\nu > 0$ is the diffusion coefficient. We assume that the equation is supplemented with the homogeneous Dirichlet boundary conditions and a suitable initial condition.

Because the unknown is scalar-valued, we will use $u$ instead of $\bu$ for the results pertaining to the Burgers equation.

The clear similarities to the NSE momentum equation (absent the pressure terms) mean that the variational equation to be solved for the Burgers equation is the same as for the standard NSE G-ROM described in \eqref{eq: gromu} cast in one spatial dimension. 
Because of these similarities with the NSE and its relatively easier numerical treatment, the Burgers equation is often used as a first numerical testbed for various 
reduced order modeling ideas and beyond, see, e.g.,~\cite{ahmed2018stabilized,sanfilippo2023approximate,KV99,KV01,Iliescu_al18}. 

To test the results of Theorems \ref{thm:verifiability} and \ref{thm:LimitConsistency} in a challenging setting, we will consider this problem with a relatively small viscosity coefficient $\nu$, for which the solution develops sharp gradients known as viscous shocks.

The semi-discretized problem to be solved for the S-ROM and L-ROM 
is the same as \eqref{eq:NSE_ODE_semi-discretized} of Section \ref{sec:implementation}, replacing $Re$ with $\nu$. Note that, because of the homogeneous boundary conditions, there is no  
zeroth mode to subtract. 
Thus, the terms $\ua_0$, $\wtua_0(\bw_r)$, $B_1$, $B_2$, and $\ub_0$ in \eqref{eq:NSE_ODE_semi-discretized} are absent here. 
We solve the S-ROM and L-ROM equations using a semi-implicit backward Euler scheme.

\begin{figure}[!htb]
    \centering
     \includegraphics[width=0.9\textwidth]{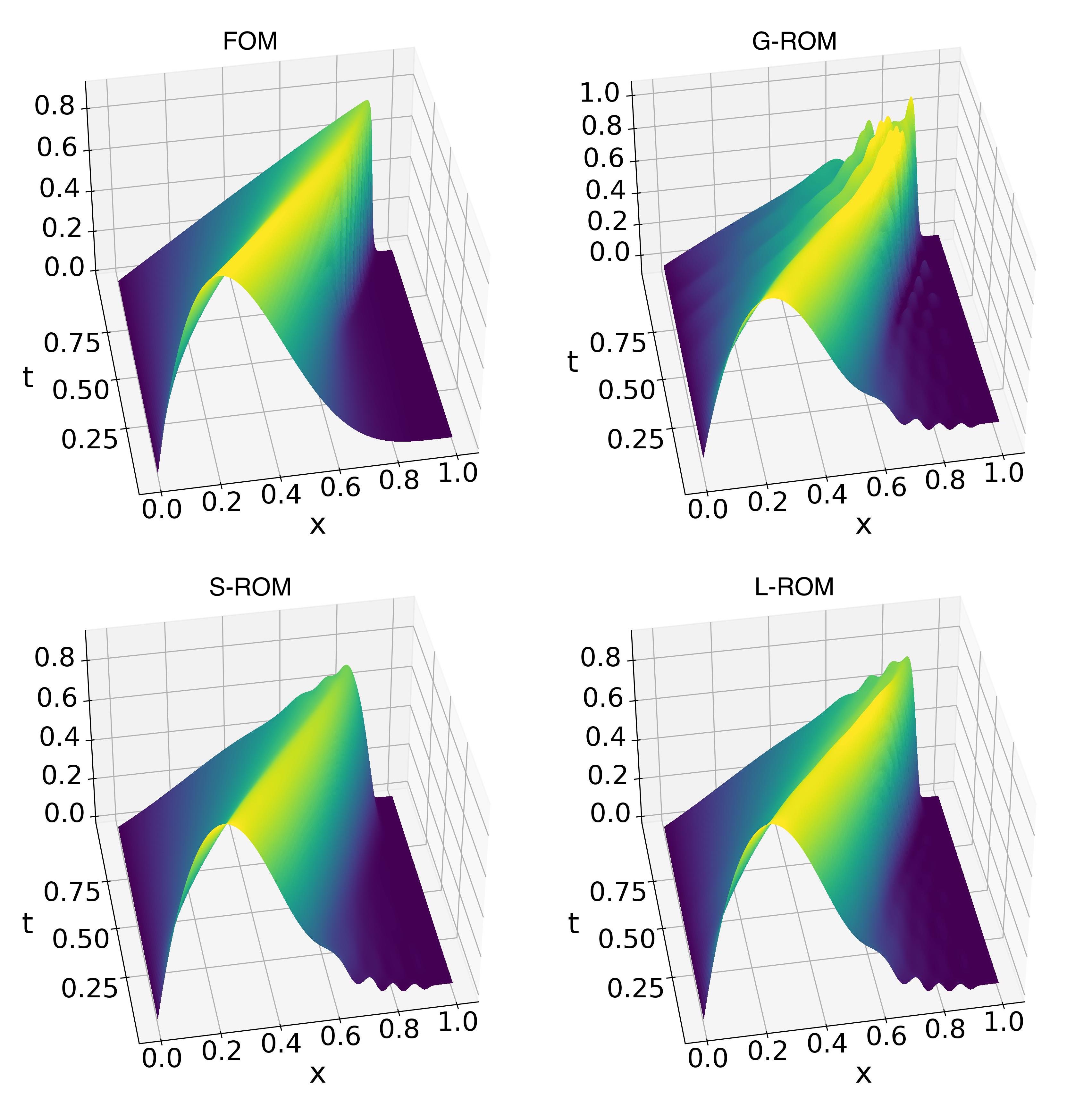}
    \caption{The Burgers equation at $\nu = 2 \times 10^{-3}$. Predicted value of $u$ comparison between the G-ROM, S-ROM, and L-ROM at $r = 10$ with $\delta = 0.04$. 
}
\label{fig:burgers_FOM_ROM_results_3D} 
\end{figure}

For all numerical results reported in this section, the time interval of simulation is set to be $[0,1]$. 
The initial condition is chosen to be 
a smooth profile $u(x,0) = 3\sin(\pi x) (1-x)^3$. By choosing a small viscosity parameter of $\nu = 2 \times 10^{-3}$, we can see in 
Fig.~\ref{fig:burgers_FOM_ROM_results_3D} 
that the FOM solution displays a rapidly steepening wavefront. 
The FOM is a 
finite element model consisting of 1000 evenly spaced linear elements ($h = 10^{-3}$) and a time step $\Delta t = 0.5 \times 10^{-3}$. 
The chosen precision turns out to be sufficient to resolve well the sharp spatial gradients developed during the evolution of the solution profiles. 
From the FOM data, we extract 50 evenly spaced (in time) snapshots to construct the 
POD ROM basis $\{ \varphi_i\}_{i = 1}^r$. All 
the ROMs thus have maximum dimension $d = 50$, and are solved with the same time step as the FOM for comparison purposes.

In Fig.~\ref{fig:burgers_FOM_ROM_results_3D}, we also 
illustrate the various ROM solutions at a dimension of $r = 10$. The G-ROM for this dimension clearly displays nonphysical oscillations and diverges significantly from the FOM data. In contrast, both the S-ROM and L-ROM are, 
as expected, able to recover a much smoother profile. We note that, for this problem, choosing the same $\delta$ value for both the S-ROM and L-ROM leads to a more diffused final solution for S-ROM, whereas the L-ROM maintains a sharper gradient more similar to the FOM. 

\subsubsection{Verifiability}

We now look to demonstrate the scaling result of Theorem \ref{thm:verifiability}. Fig.~\ref{fig:burgers_constant_delta_verifiability} plots the ROM error $\ROMError$ (\ref{eq:error_l2}) and closure error $\ClosureError$ (\ref{eq:closure_l2}) over all $r$ values for both the S-ROM and L-ROM using the mesh size $h$ as the value of $\delta$. 
We note that similar lengthscale scalings are used to construct LES models~\cite{berselli2006mathematics,sagaut2006large}.
Other choices can (and probably should) be used in the ROM setting, e.g., ROM lengthscales that adapt as we change the ROM dimension (see the energy-based ROM lengthscale~\cite{mou2023energy}).
The results generally show that, as $\ClosureError$ decreases, so does $\ROMError$ for both the S-ROM and L-ROM.

Fig.~\ref{fig:burgers_constant_delta_verifiability_regression} demonstrates the scaling results over the $r$ values in the range of $[15,35]$, for which both the ROM error and closure error are decreasing (i.e., the ROM has not reached saturation and the typically variable results of low $r$ values are excluded). Both models provide a linear regression (LR) with a slope somewhat higher than the expected value of 1.

\begin{figure}[!hbt]
    \centering
    \begin{subfigure}{0.49\columnwidth}
    \caption{S-ROM with $\delta = h = 0.001$}
     \includegraphics[width=1\textwidth]{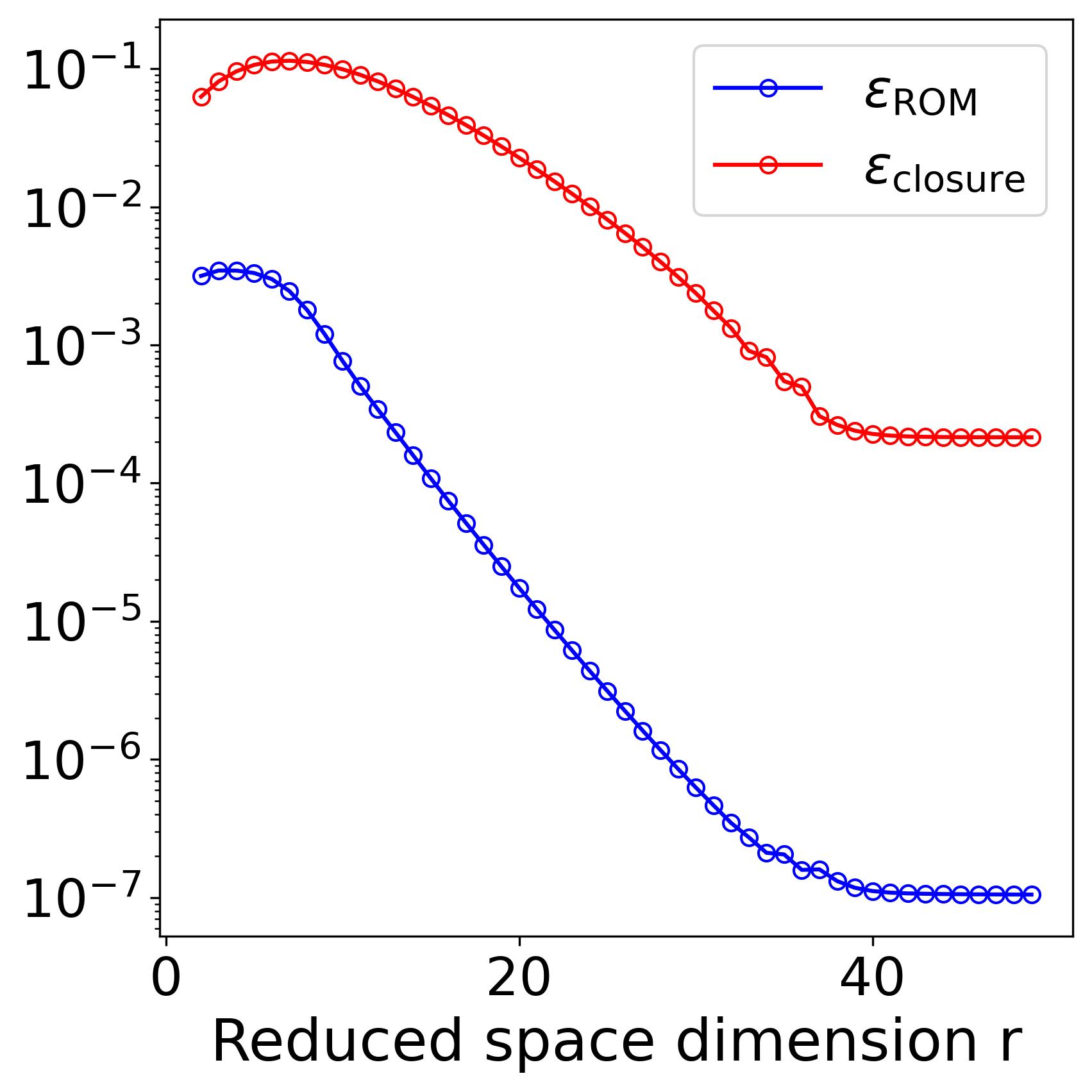}
     
    \end{subfigure}
    \begin{subfigure}{0.49\columnwidth}
    \caption{L-ROM with $\delta = h = 0.001$}
     \includegraphics[width=1\textwidth]{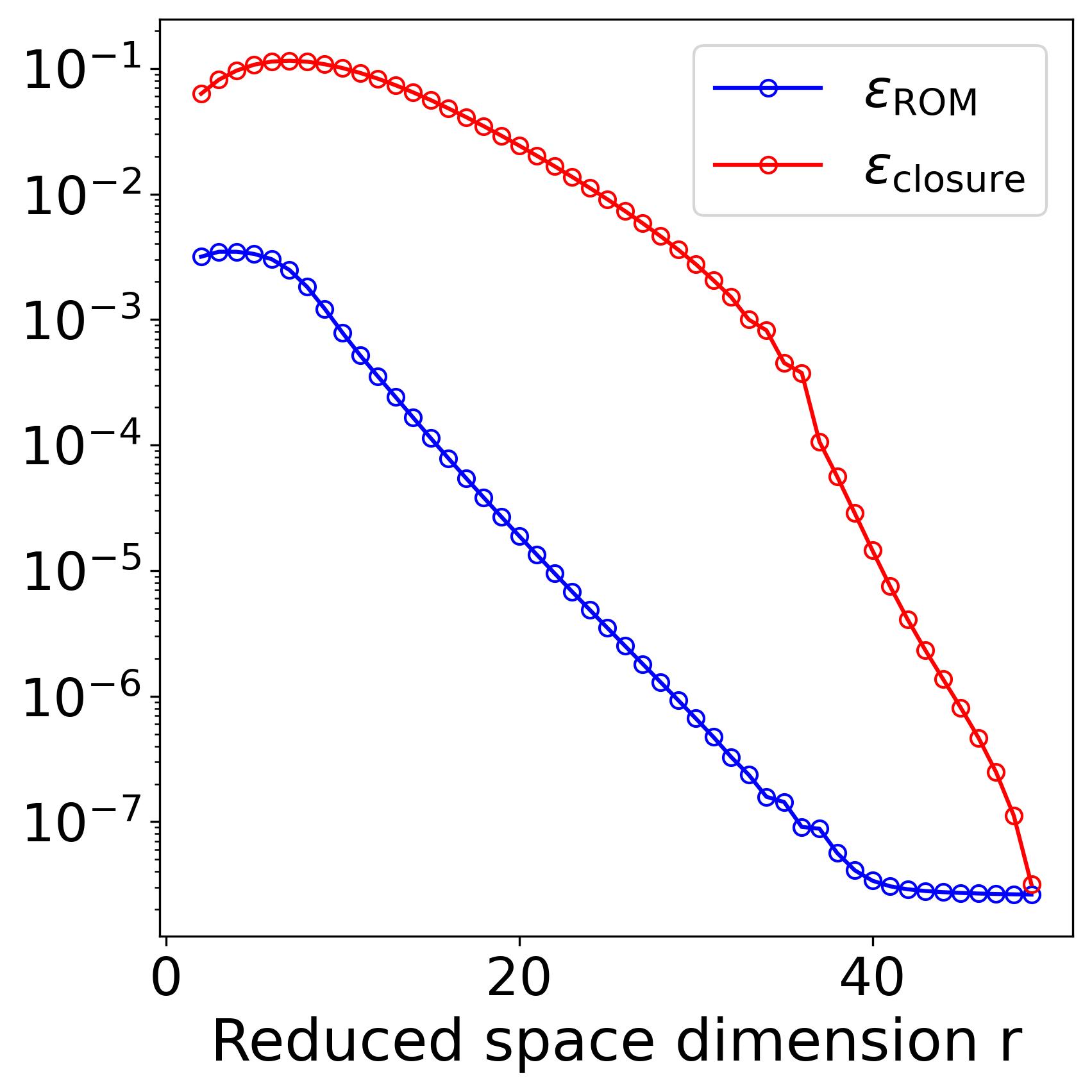}
    \end{subfigure}
    \caption{The Burgers equation at $\nu = 2 \times 10^{-3}$. 
    The ROM error $\ROMError$ and the closure error $\ClosureError$ with respect to different $r$ values for the S-ROM and L-ROM with $C_S=1$ and $\delta =0.001$.}
    \label{fig:burgers_constant_delta_verifiability}
\end{figure}

\begin{figure}[!hbt]
    \centering
    \begin{subfigure}{0.49\columnwidth}
    \caption{S-ROM with $\delta = h = 0.001$}
     \includegraphics[width=1\textwidth]{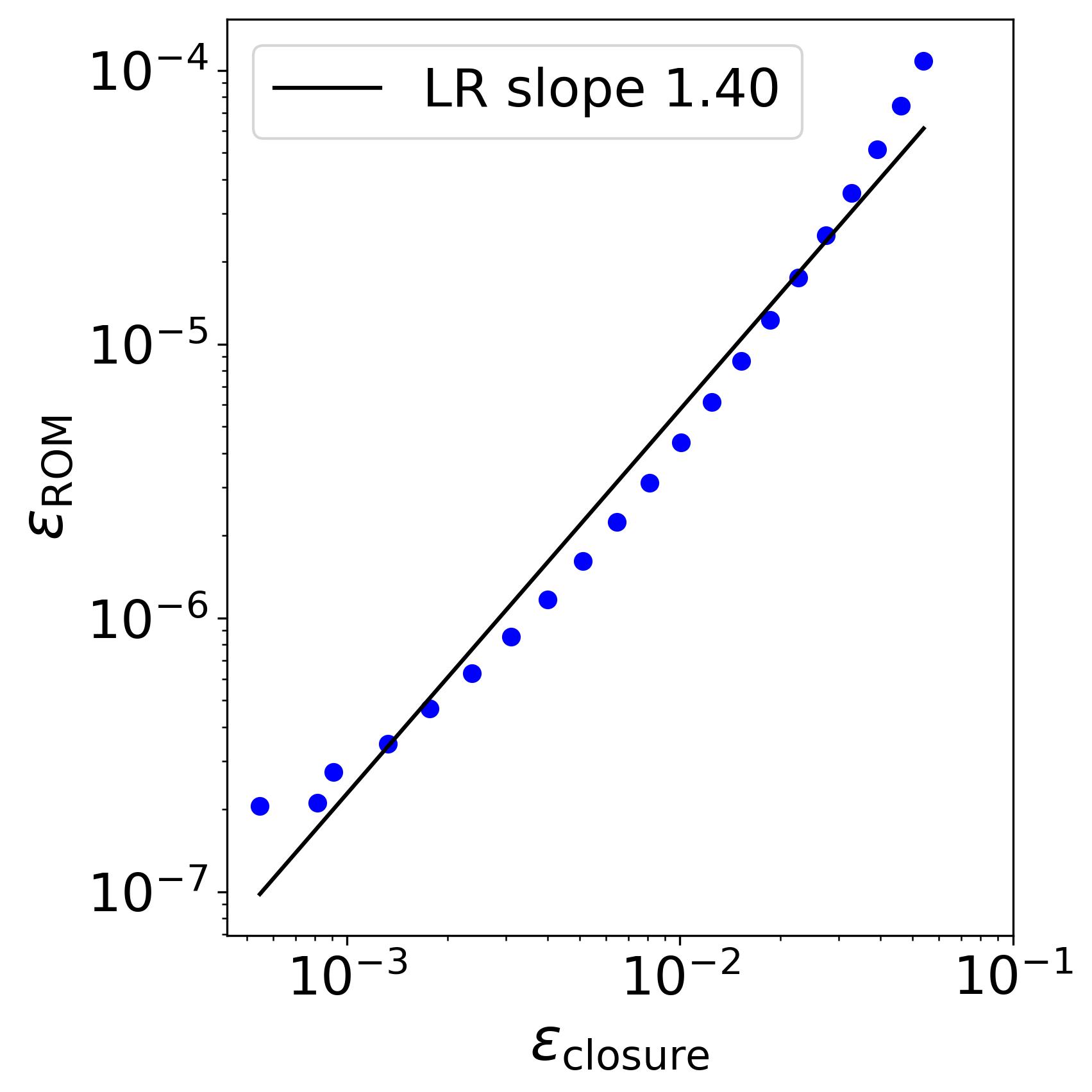}
     
    \end{subfigure}
    \begin{subfigure}{0.49\columnwidth}
    \caption{L-ROM with $\delta = h = 0.001$}
     \includegraphics[width=1\textwidth]{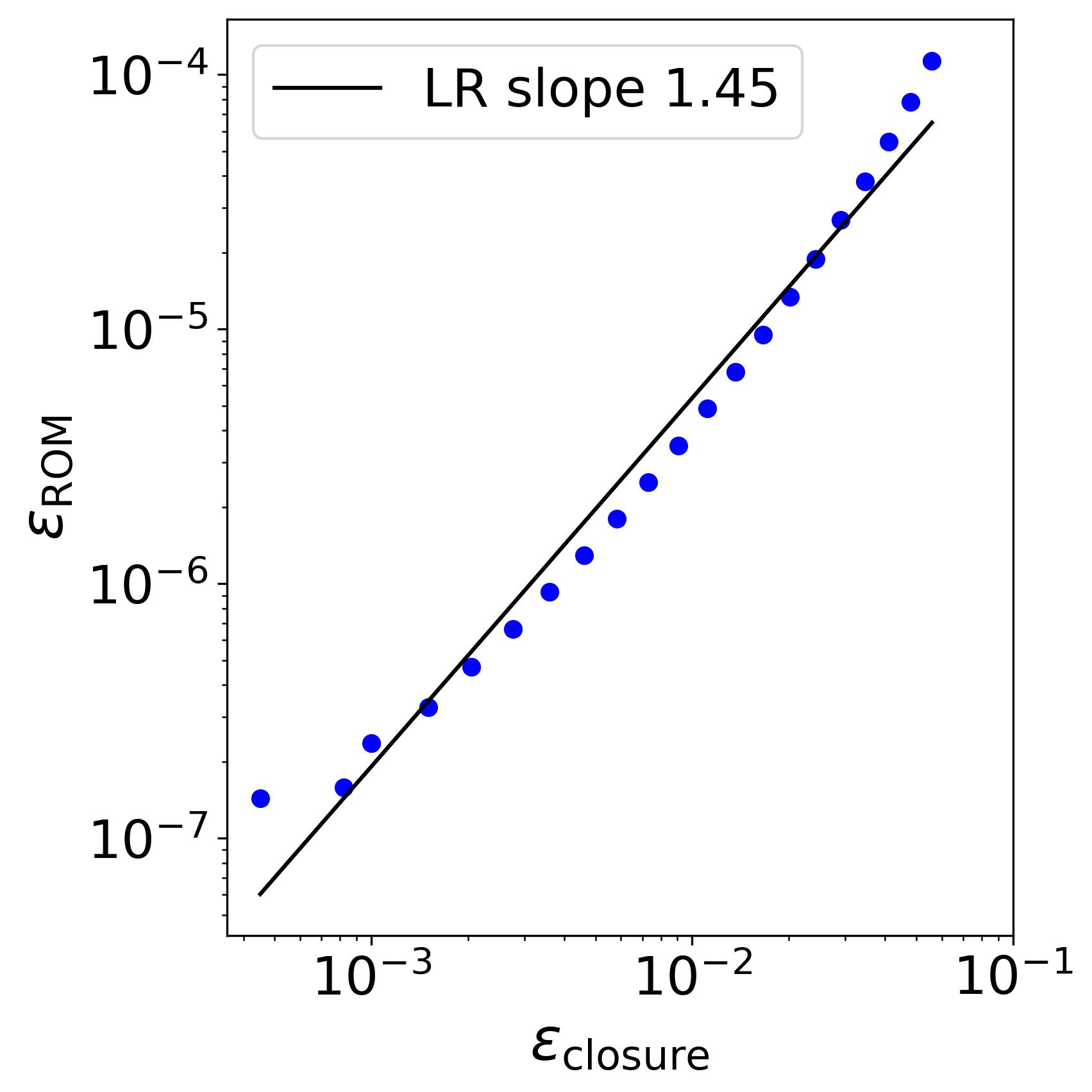}
    \end{subfigure}
    \caption{The Burgers equation at $\nu = 2 \times 10^{-3}$. The behavior of the ROM error $\ROMError$ with respect to the closure error $\ClosureError$. 
    Values shown are for $r = 15, \dots, 35.$ }
    \label{fig:burgers_constant_delta_verifiability_regression}
\end{figure}

\subsubsection{Limit Consistency}

We now proceed to test the limit consistency results of Theorem \ref{thm:LimitConsistency} and Theorem \ref{thm:Thm3}. We recall from Definition \ref{definition:limit-consistency} that an LES-ROM is limit consistent if the LES-ROM solution $w_r$ approaches the d-dimensional G-ROM solution (``true solution") $u_d$ as $\delta \to 0$ and as $r$ increases from $1$ to $d$. Theorem \ref{thm:LimitConsistency} states that $\eplimit $\eqref{eq:limit_left-hand side}, the error between $w_r$ and $u_d$, should be bounded by the consistency error, $ \etalimit$ \eqref{eq:limit_right-hand side}. 
To investigate this numerically, we take a sequence $\{(r_i, \delta_i)\}_{i=1}^8$ such that as $r_i$ approaches $d$, $\delta_i$ approaches 0. 
Specifically, we chose the sequence 
$\{ (15, 1 \times 10^{-2}) , (20, 4.6 \times 10^{-3}), (25, 2.2 \times 10^{-3}) , (30,1 \times 10^{-3})), (35, 4.6 \times 10^{-4})),(40, 2.1 \times 10^{-4})), (45,1 \times 10^{-4}))\}$. Fig.~\ref{fig:burgers_limit_consistency_new} demonstrates the result of this choice, and shows that both the S-ROM and L-ROM errors decay at similar rates as $\etalimit$ shrinks. 

\begin{figure}[!hbt]
    \centering
    \begin{subfigure}{0.49\columnwidth}
    \caption{S-ROM}
     \includegraphics[width=1\textwidth]{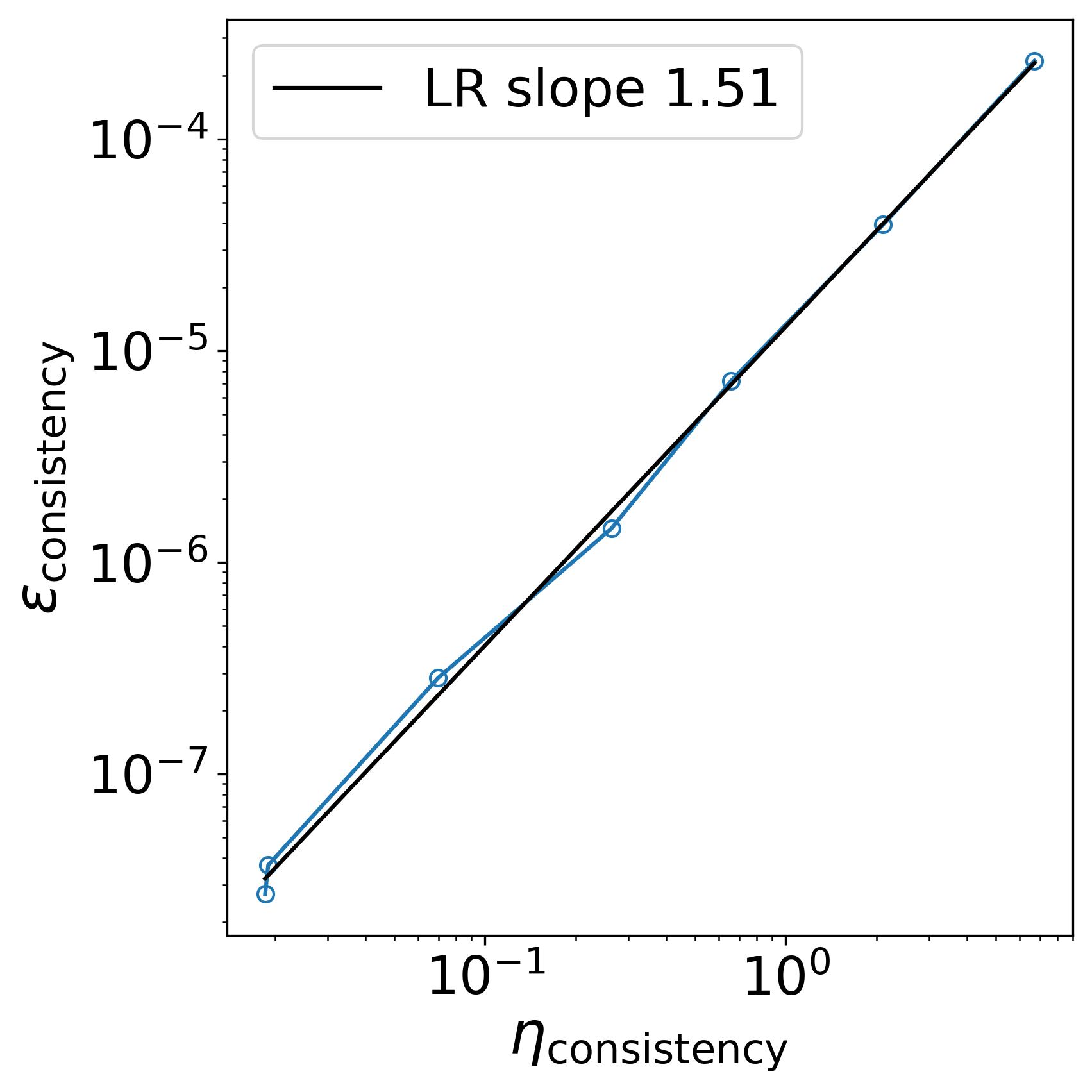}
     
    \end{subfigure}
    \begin{subfigure}{0.49\columnwidth}
    \caption{L-ROM}
     \includegraphics[width=1\textwidth]{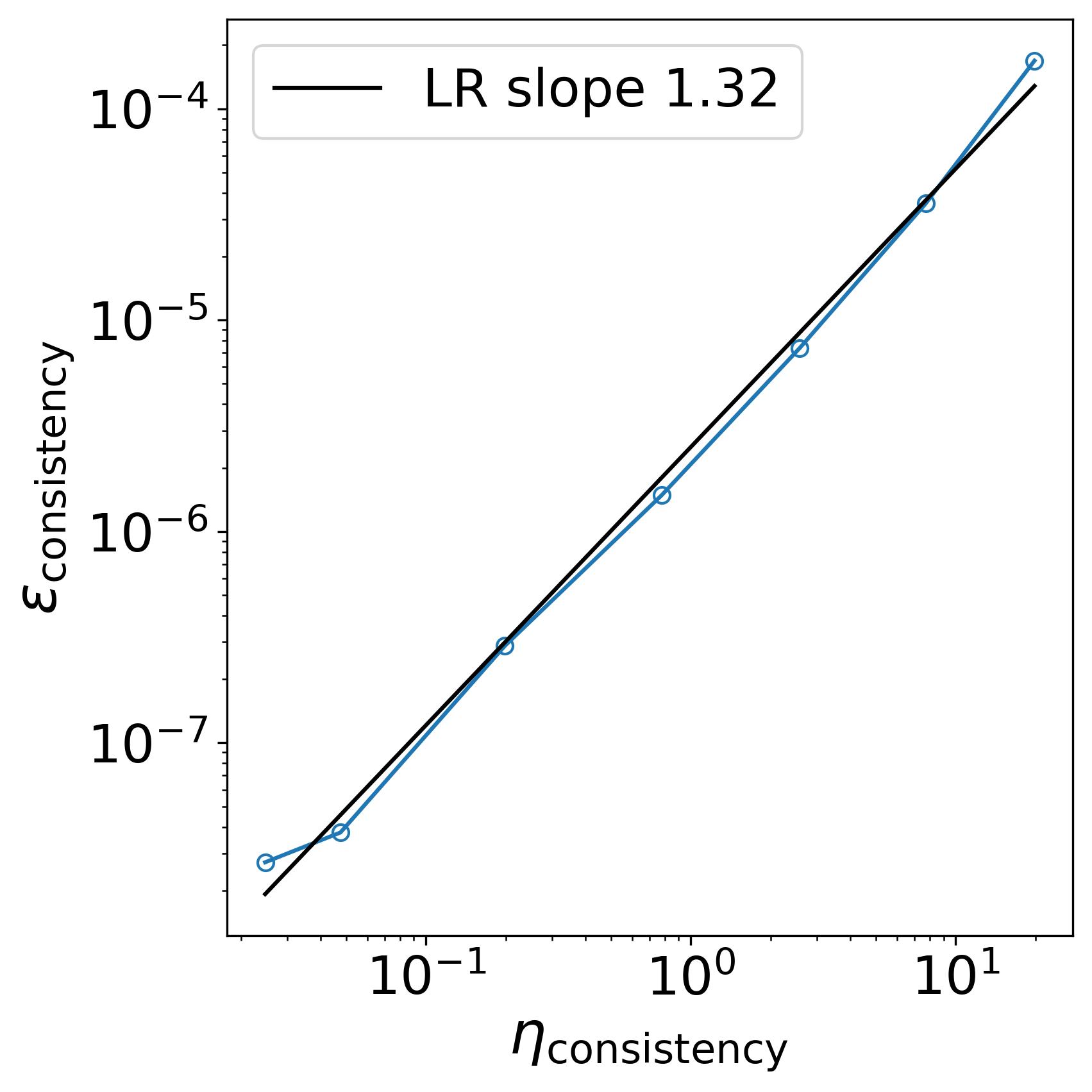}
    \end{subfigure}
    \caption{The Burgers equation at $\nu = 2 \times 10^{-3}$, demonstrating the limit consistency of the S-ROM and L-ROM.
    }
    \label{fig:burgers_limit_consistency_new}
\end{figure}

Finally, following  Remark \ref{remark:delta-convergence}, equation \eqref{eq: L-ROM_LC} in Theorem \ref{thm:Thm3} suggests that $\Delta t\sum_{n=1}^M \| \tau^{\mathrm{ROM}}(u^n_r) \|$ should scale like $\mathcal{O}(\delta^\mu)$. Recall that we have chosen $\mu = 10/3$ for our implementation of the L-ROM, and that $\mu = 2$ for the S-ROM. Lemma \ref{lem: wr2ur} bounds the quantity $\| u_r - w_r\|$  by the norm of the closure term $\tau^{\mathrm{ROM}}$, suggesting that both $\tau^{\mathrm{ROM}}$ and $\| u_r - w_r\|$ should decay at the same rate as $\delta \to 0$.  
We demonstrate this rate in Fig.~\ref{fig:burgers_limit_consistency} for a fixed value of $r = 10$ and taking 10 evenly spaced $\delta$ values between $10^{-4}$ and $10^{-2}$. Noting that the norms and thus the expected rates are squared, we see that the S-ROM results demonstrate a rate of $\mathcal{O}(10^{3.83})$, and the L-ROM demonstrates a rate of $\mathcal{O}(\delta^{6.66})$. The S-ROM results are slightly under the expected theoretical rate of 4, while the L-ROM results are accurate to the theory up to two decimal places. Further results are omitted for brevity, but the authors have found that similar results hold for this problem over varied values of $r$.

\begin{figure}[!hbt]
    \centering
    \begin{subfigure}{0.49\columnwidth}
    \caption{S-ROM with $r = 10$}
     \includegraphics[width=1\textwidth]{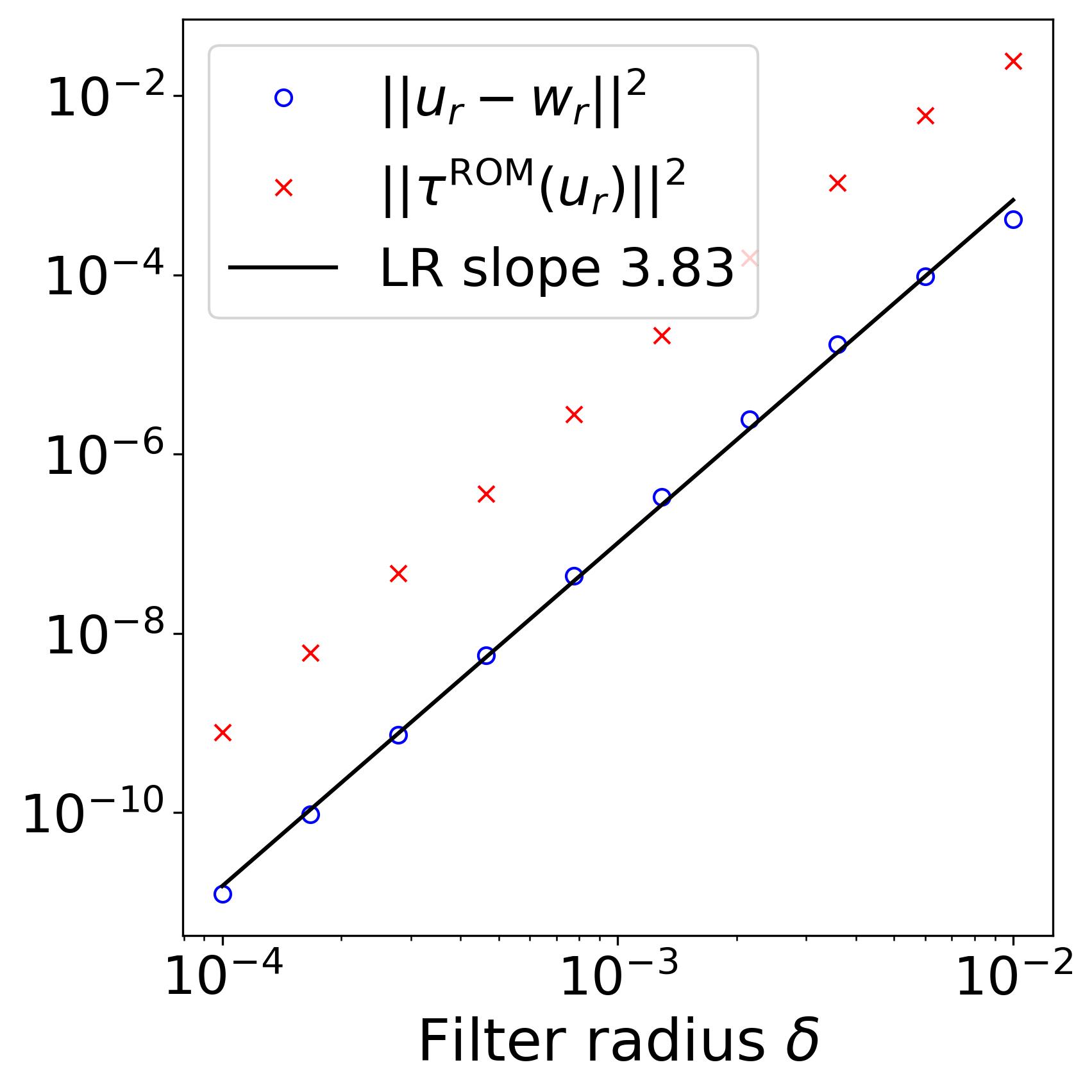}
     
    \end{subfigure}
    \begin{subfigure}{0.49\columnwidth}
    \caption{L-ROM with $r = 10$}
     \includegraphics[width=1\textwidth]{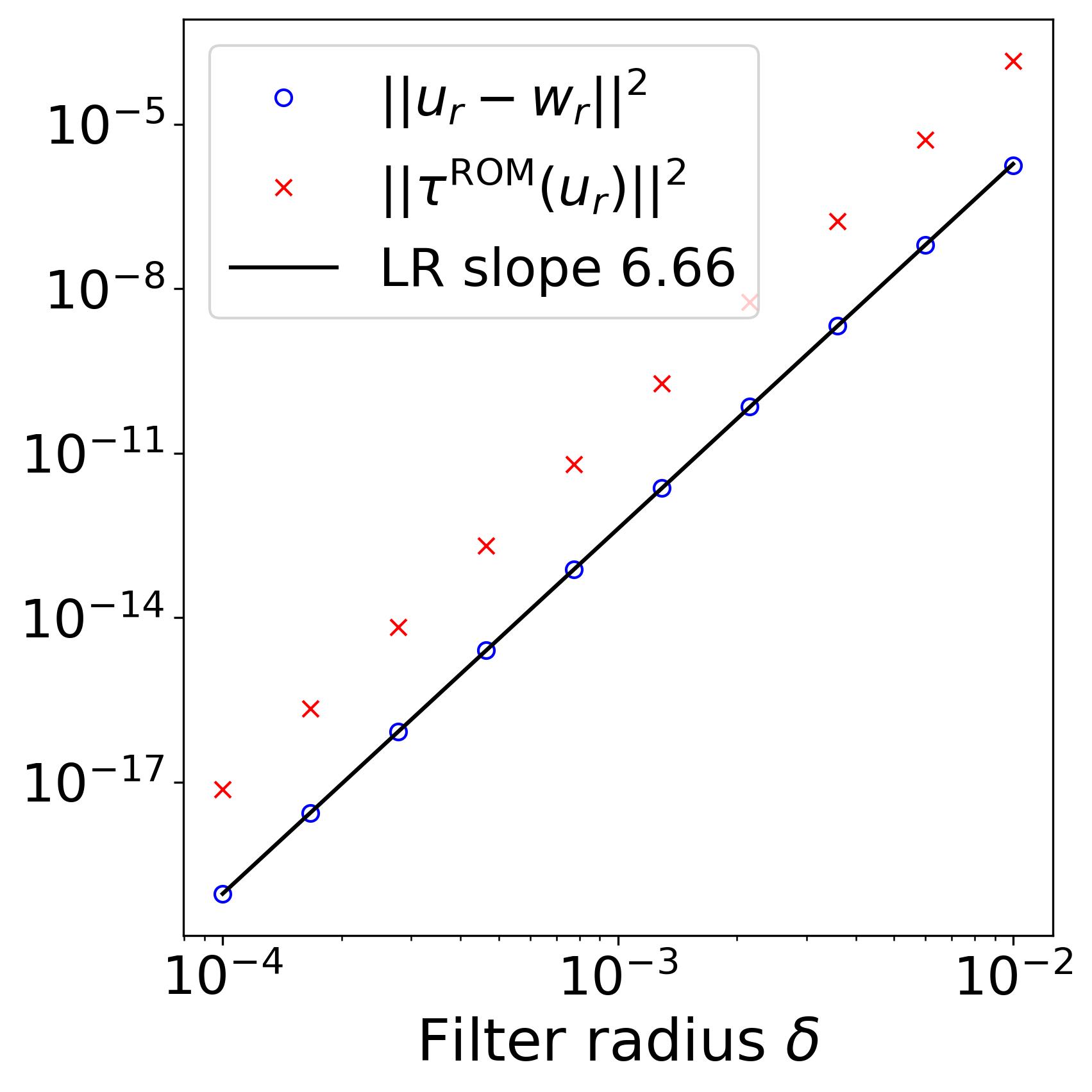}
    \end{subfigure}
    \caption{The Burgers equation at $\nu = 2 \times 10^{-3}$. The regression result demonstrates the rate suggested by Theorem \ref{thm:Thm3}. 
    }\label{fig:burgers_limit_consistency}
\end{figure}
\subsection{2D Lid-Driven Cavity}
\label{subsec:2dldc}

Our next 
{test case} is the 2D lid-driven cavity problem at $\rm Re=15,000$. 
We consider  
{the direct numerical simulation} to be the 
{FOM}. 
A detailed description 
of the FOM setup for this problem can be found in \cite{kaneko2020towards}.
As demonstrated in \cite{tsai2023accelerating}, this model problem requires more than $60$ POD modes for the G-ROM (\ref{eq: gromu-r}) to accurately reconstruct the solutions and 
{quantities of interest.}

To construct the S-ROM and L-ROM, 
we first simulate the FOM until the solutions reach a statistically steady state. We then collect $M=2001$ snapshots $\{\bu^n :=
\bu(\bx,t^n)-\bphi_0\}^M_{n=1}$ 
in the time
interval 
$[6000,~ 6010]$ with a sampling time of $\Delta t_s=0.005$. 
\begin{figure}[!ht]
    \centering
     \includegraphics[width=0.75\textwidth]{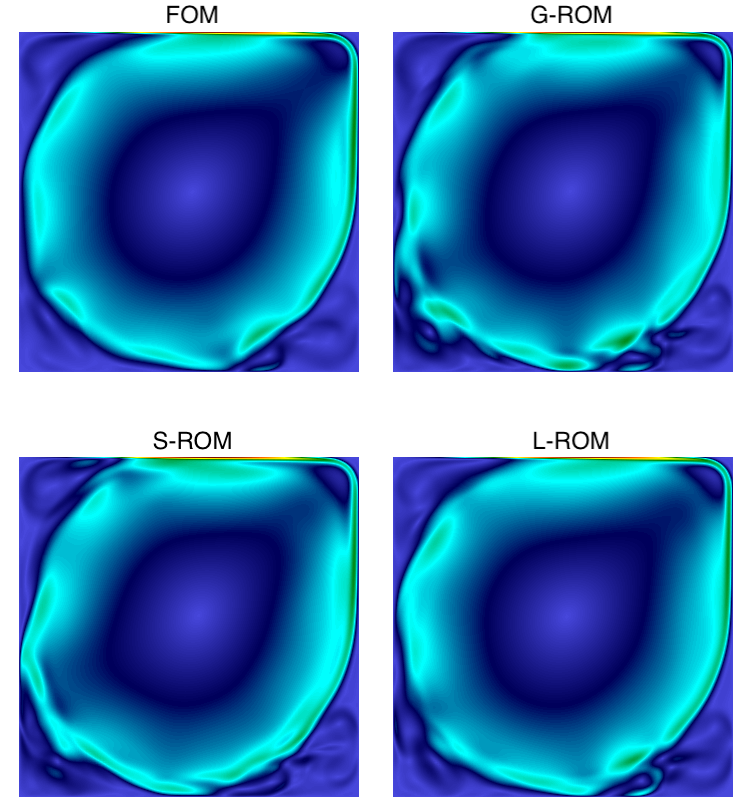}
    \caption{2D lid-driven cavity at $\rm Re=15,000$. 
    Predicted velocity magnitude comparison between the G-ROM, 
    S-ROM, and L-ROM with $\delta=0.006$, along with the FOM velocity magnitude.
}
    \label{fig:2dldc_sol_compare}
\end{figure}
The reduced basis functions $\{\bphi_i\}^r_{i=1}$ are constructed by applying POD
to the snapshot set.
We note that the zeroth
mode, $\bphi_0$, is set to be the FOM velocity field 
{at the initial time instance, i.e., at} $t=6000$. 
{We also note that, although the verifiability and limit consistency results proved in this paper do not directly apply to the case of ROMs that use a centering trajectory (i.e., the zeroth mode), we believe that these theoretical results can be extended to this setting.}

In Fig.~\ref{fig:2dldc_sol_compare}, we show the FOM solution at $t=6025$ and the predicted solution by the G-ROM, the S-ROM, and the L-ROM with $r=8$ and $\delta=0.006$. The G-ROM displays nonphysical oscillations and diverges from the FOM data. The S-ROM with $\delta=0.006$ is overly dissipative and its solution is much smoother than the FOM solution. On the other hand, the L-ROM with $\delta=0.006$ is able to stabilize the flow without being overly dissipative.

\subsubsection{{Verifiability}}

We now look to demonstrate the scaling result of Theorem \ref{thm:verifiability}. 

In Fig.~\ref{fig:2dldc_error_vs_cerror_1}, {we vary the reduced space dimension, $r$, and} we examine how $\ROMError$ (\ref{eq:error_l2}) and $\ClosureError$ (\ref{eq:closure_l2}) vary {for the S-ROM and L-ROM}. 
The filter radius $\delta$ is set to $0.001$ for both the S-ROM and L-ROM.

\begin{figure}[!ht]
    \centering
    \begin{subfigure}{0.49\columnwidth}
    \caption{S-ROM with $\delta=0.001$}
     \includegraphics[width=1\textwidth]{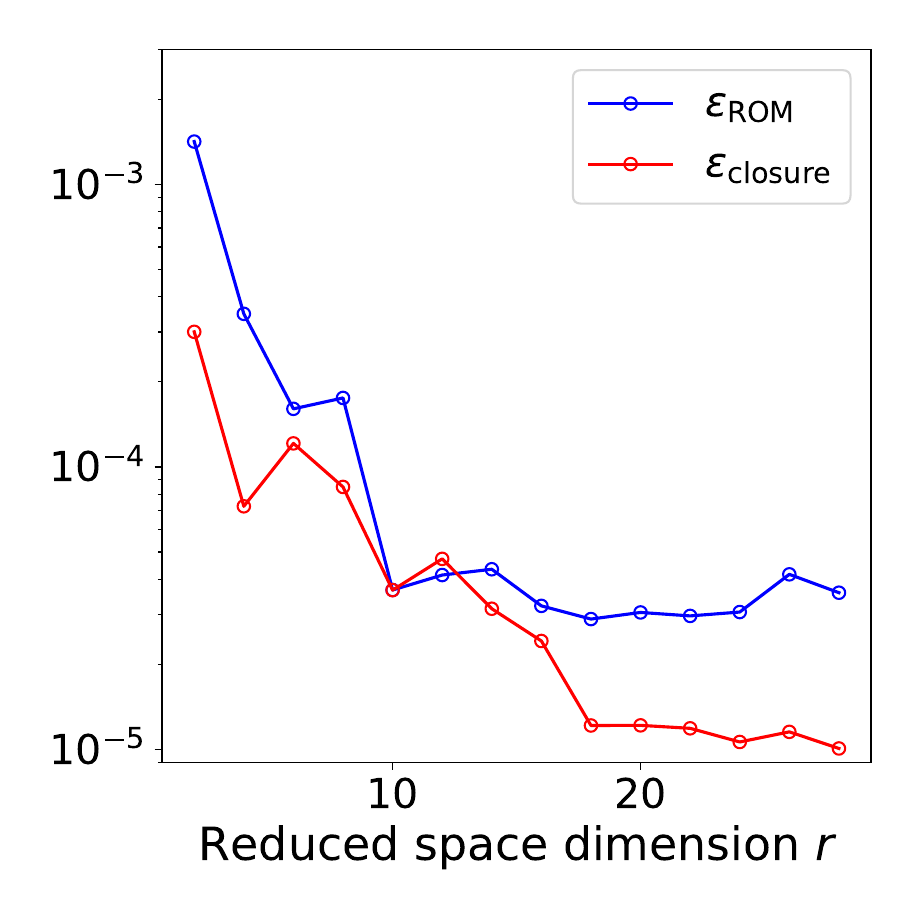}
    \end{subfigure}
    \begin{subfigure}{0.49\columnwidth}
    \caption{L-ROM with $\delta=0.001$}
     \includegraphics[width=1\textwidth]{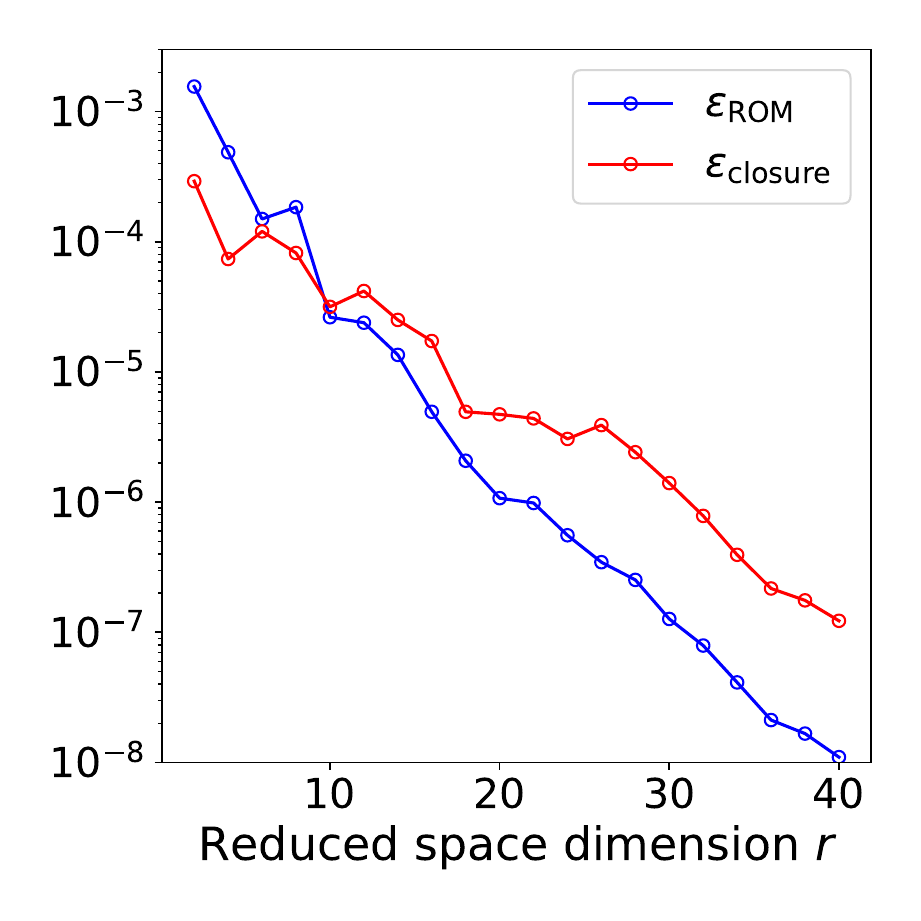}
    \end{subfigure}
    \caption{2D lid-driven cavity at $\rm Re=15,000$. 
    The ROM error $\ROMError$ and the closure error $\ClosureError$ with respect to different $r$ values for the S-ROM and L-ROM with $C_S=1$ and $\delta =0.001$. 
    }
    \label{fig:2dldc_error_vs_cerror_1}
\end{figure}

The results generally show that, as $\ClosureError$ decreases, so does $\ROMError$ for both the S-ROM and L-ROM.
\begin{figure}[!ht]
    \centering
    \begin{subfigure}{0.49\columnwidth}
    \caption{S-ROM with $\delta=0.001$}
    \includegraphics[width=1\textwidth]{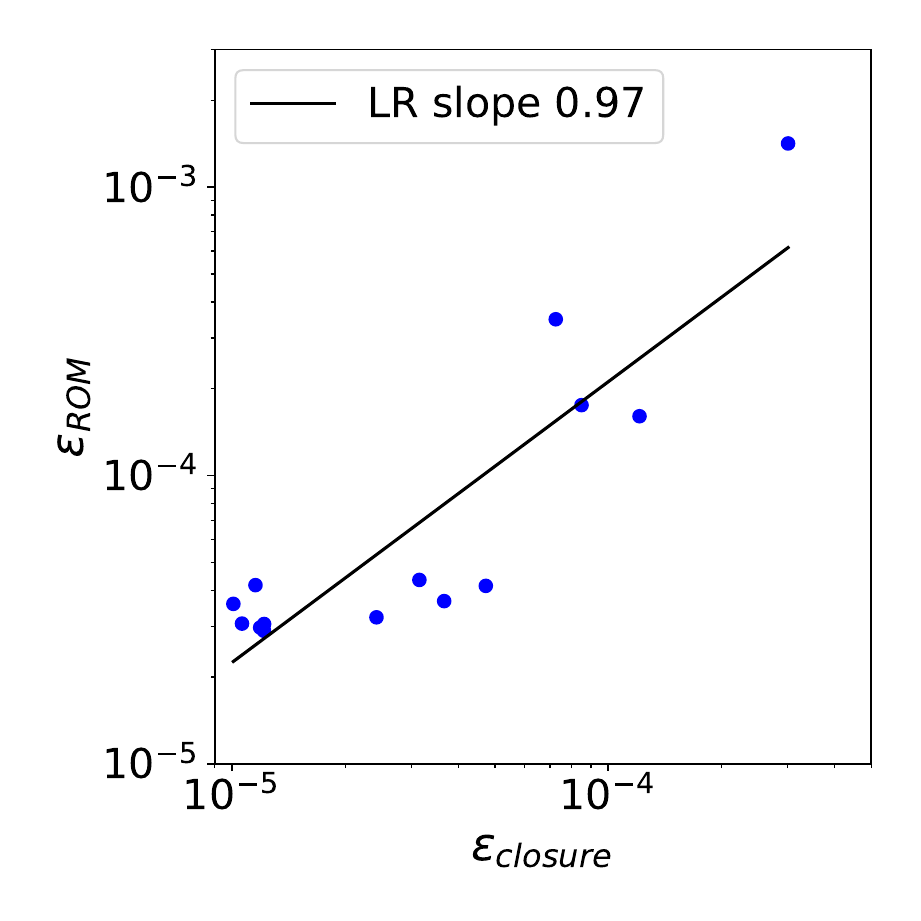}
    \end{subfigure}
    \begin{subfigure}{0.49\columnwidth}
    \caption{L-ROM with $\delta=0.001$}
    \includegraphics[width=1\textwidth]{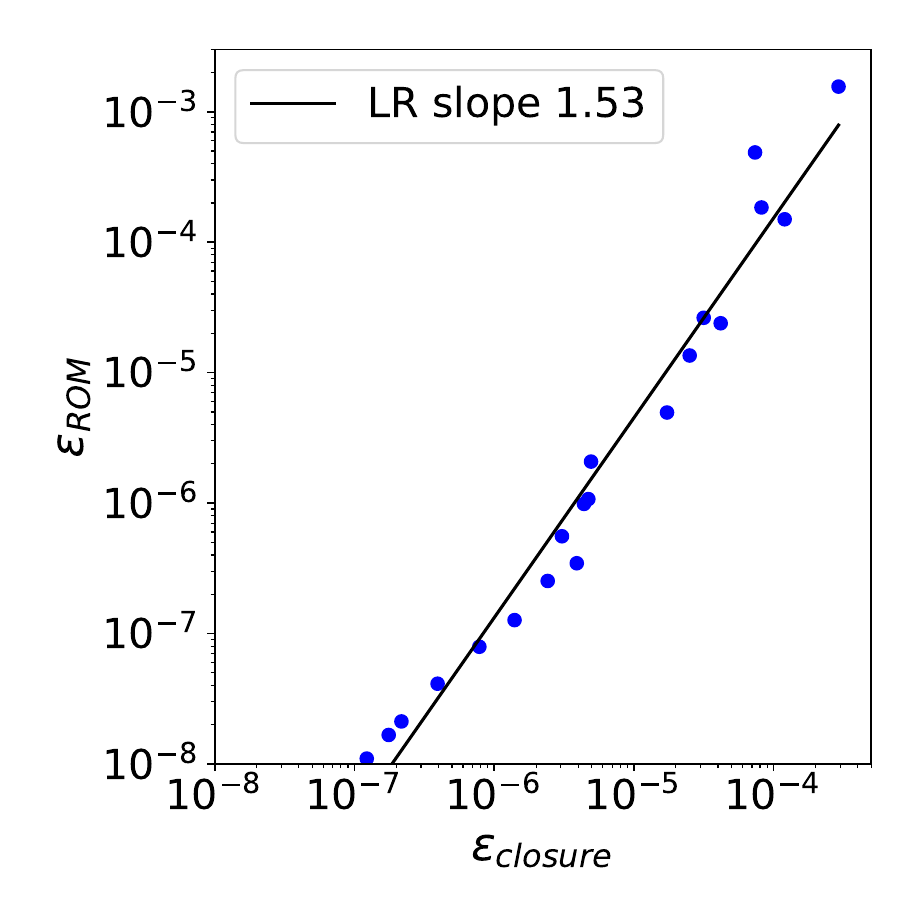}
    \end{subfigure}
    \caption{2D lid-driven cavity at $\rm Re=15,000$. 
    The behavior of the ROM error $\ROMError$ with respect to the closure error $\ClosureError$. 
    }
    \label{fig:2dldc_verif_1}
\end{figure}

Furthermore, we monitor the 
decay rate of $\ROMError$ with respect to $\ClosureError$ as $r$ varies for both the S-ROM and L-ROM.
Specifically, we examine the relation (\ref{eq:relation}) for a fixed $\delta=0.001$ while $r$ is varied. For this purpose, in Fig.~\ref{fig:2dldc_verif_1}, we display the slope obtained from linear regression 
of the $(\ROMError, \ClosureError)$ data pairs.
The LR slope for the S-ROM is computed using data for $r=2,4,\ldots,28$, whereas for the L-ROM, it is based on $r=2,4,\ldots,40$. 
The results show that (\ref{eq:relation}) holds with $\alpha \approx 1$ for the S-ROM, and with $\alpha > 1$ for the L-ROM.

\begin{remark}
    We note that the rate of verifiability depends on the length of the training time interval.
    In particular, we were able to demonstrate the verifiability of the S-ROM and L-ROM for different time interval lengths that are less than $20$. For lengths that are greater than $20$, we found that the slope $\alpha$ was below $1$.
    This is because the ROM error (\ref{eq:error_l2}) is a pointwise measurement in time and as the length of the time interval increases, the ROM solution trajectory could be very different from the FOM solution due to  
    {the chaotic character of the underlying system. This} 
    affects the quality of the ROM error (\ref{eq:error_l2}). 
\end{remark}

\subsubsection{{Limit Consistency}}

As shown in Fig.~\ref{fig: LC}, the LES-ROM is limit consistent if the LES-ROM solution $\bw_r$ approaches the $d$-dimensional G-ROM solution $\bfu_d$ as $\delta \rightarrow 0 $ and $r$ increases from $1$ to $d$. In Theorem~\ref{thm:LimitConsistency}, we prove that 
the error between 
$\bw_r$  
and $\bfu_d$ 
is bounded by $\etalimit$ (\ref{eq:limit_right-hand side}), {if the LES-ROM's closure term satisfies the mean dissipativity {condition}~(\ref{eq: mean_dis}) and Assumption~\ref{as: assumption 1}. 
In Theorem~\ref{thm:Thm3}, we prove that the S-ROM and L-ROM are limit consistent.}

To numerically illustrate the limit consistency, 
we evaluate $\eplimit$ and $\etalimit$ at a sequence of  
$\{(r_i,\delta_i)\}^{N}_{i=1}$ values, and show that as $\etalimit$ decreases, so does $\eplimit$. 
\begin{figure}[!ht]
    \centering
    \begin{subfigure}{0.49\columnwidth}
    \caption{S-ROM}
     \includegraphics[width=1\textwidth]{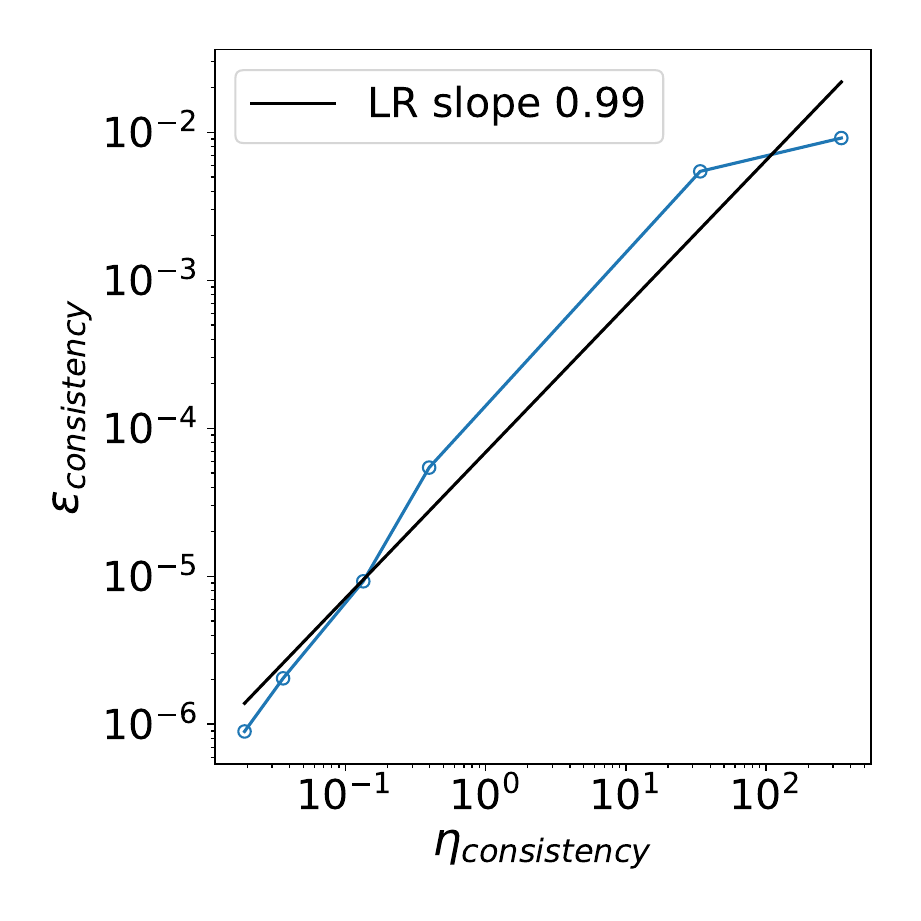}
    \end{subfigure}
    \begin{subfigure}{0.49\columnwidth}
    \caption{L-ROM}
     \includegraphics[width=1\textwidth]{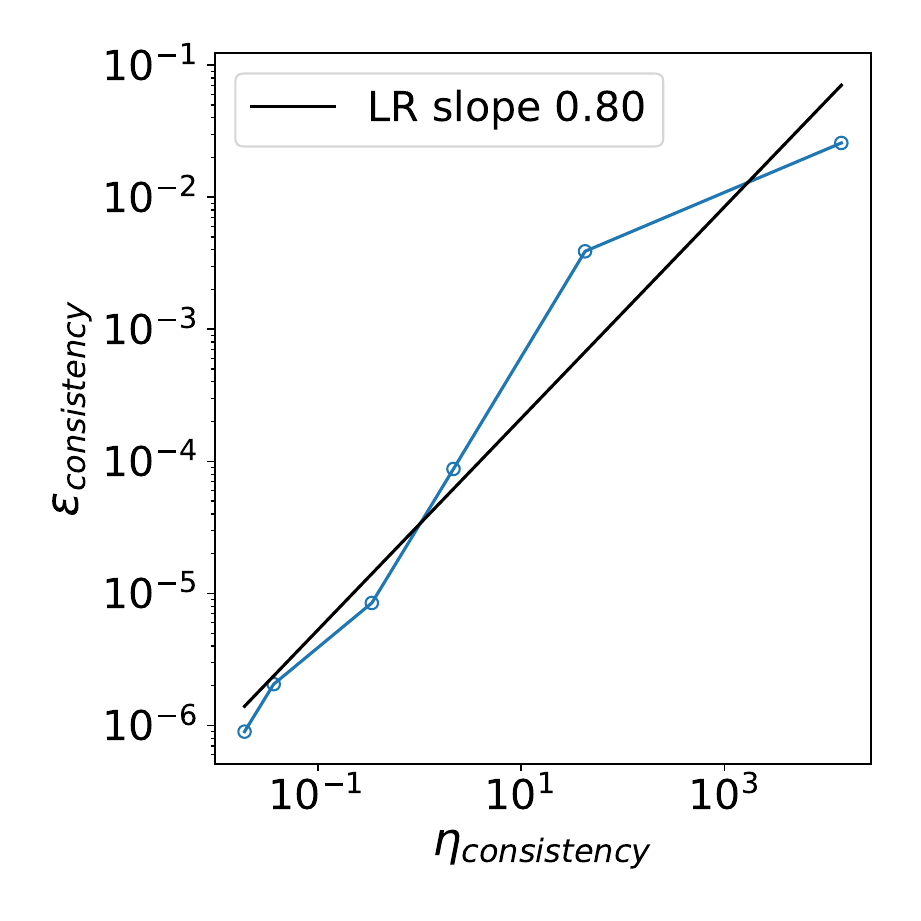}
    \end{subfigure}
    \caption{2D lid-driven cavity at $\rm Re=15,000$. 
    The limit consistency of the S-ROM and L-ROM.
    }
    \label{fig:2dldc_limit_consistency}
\end{figure}
{Specifically,} we consider the sequence $\{
(8,10^{-2}), (10,5 \times 10^{-3}), (12,10^{-3}), (16,5 \times 10^{-4}), (20,10^{-4}), (24, 5 \times 10^{-5})\}$ for the S-ROM, and the sequence $\{ 
(8,2.5 \times 10^{-2}), (10, 10^{-2}), (12,5 \times 10^{-3}), (16,2.5 \times 10^{-3}), (20, 10^{-3})$, $(24,5 \times 10^{-4})\}$ for the L-ROM,
{and plot in Fig.~\ref{fig:2dldc_limit_consistency} the corresponding $\eplimit$ and $\etalimit$ values.
We observe that 
as $\etalimit$ decreases, so does $\eplimit$ in both S-ROM and L-ROM, with a rate of $0.99$ and $0.8$, respectively. 
}

{
Theorem~\ref{thm:Thm3} and Remark \ref{remark:delta-convergence} {show that} the ROM closure, $\Delta t\sum_{n=1}^M \| \tau^{\mathrm{ROM}}(u^n_r) \|$, and the error between the LES-ROM solution, $\bw_r$, and the $r$-dimensional G-ROM solution, $\bfu_r$, should scale like $\mathcal{O}(\delta^\mu)$. }

We demonstrate this rate in Fig.~\ref{fig:2dldc_lc_rate} for a fixed $r = 10$ and taking $8$ 
$\delta$ values between $10^{-2}$ and $5\times 10^{-5}$ for the S-ROM, and $10$ $\delta$ values between $3\times 10^{-2}$ and $10^{-5}$ for the L-ROM. Noting that the norms and thus the expected rates are squared, we observe that the S-ROM results demonstrate a rate of $\mathcal{O}(\delta^{3.71})$, and the L-ROM results demonstrate a rate of $\mathcal{O}(\delta^{6.22})$, which are close to the expected theoretical rates of $4$ and $6.66$, respectively. 
\begin{figure}[!ht]
    \centering
    \begin{subfigure}{0.49\columnwidth}
    \caption{S-ROM with $r=10$}
     \includegraphics[width=1\textwidth]{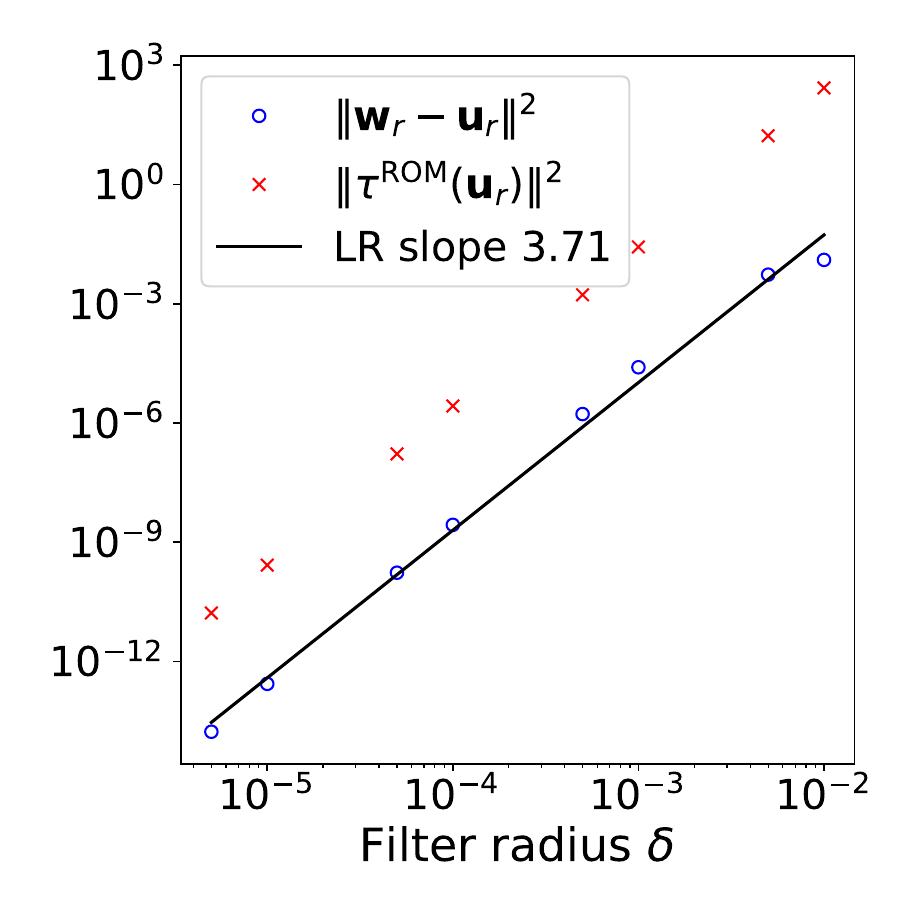}
    \end{subfigure}
    \begin{subfigure}{0.49\columnwidth}
    \caption{L-ROM with $r=10$}
     \includegraphics[width=1\textwidth]{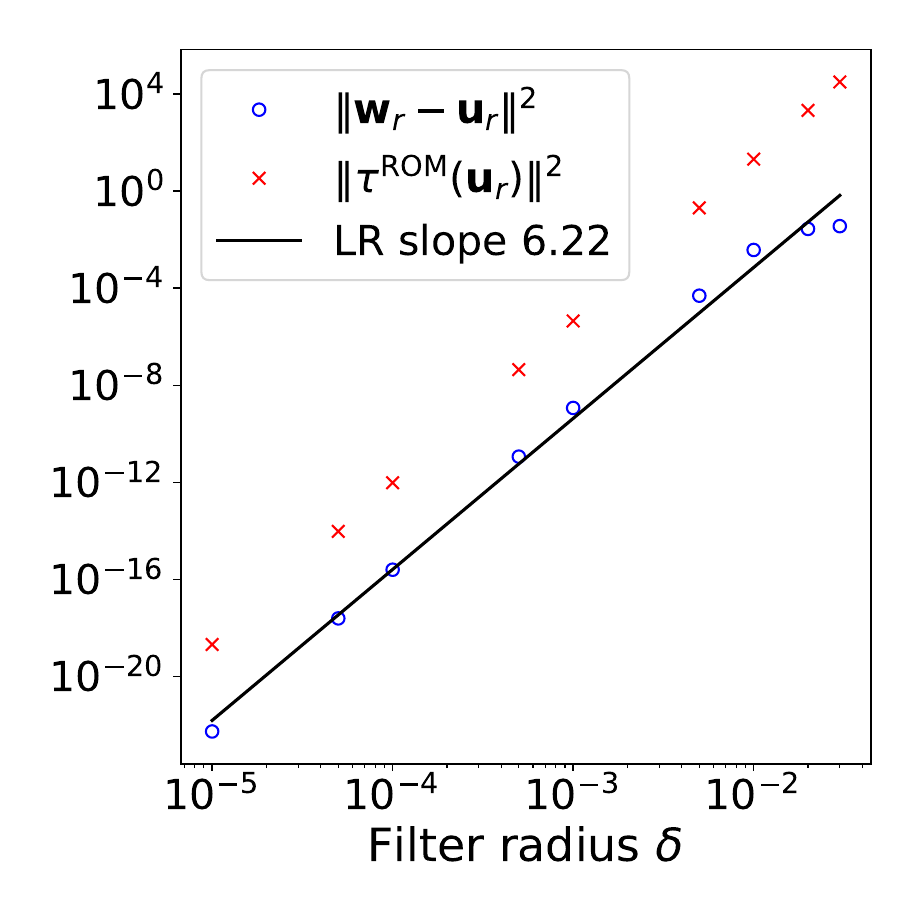}
    \end{subfigure}
    \caption{2D lid-driven cavity at $\rm Re=15,000$. The regression results demonstrate 
    the rate suggested by Theorem~\ref{thm:Thm3} and Remark~\ref{remark:delta-convergence}.} 
    \label{fig:2dldc_lc_rate}
\end{figure}

\section{Conclusions} \label{sec-Con}

In this paper, we proposed 
the Ladyzhenskaya-ROM (L-ROM), which is a new EV LES-ROM for under-resolved, convection-dominated (e.g., turbulent) flows. 
L-ROM, which is a generalization of the classical S-ROM closure, consists of a diffusion-like operator with a viscosity that depends on the Frobenius norm of the gradient of the ROM velocity.
For both the new L-ROM and 
classical S-ROM, we proved verifiability (Theorem~\ref{thm:verifiability}) and, in a discrete sense, limit consistency (Theorem~\ref{thm:LimitConsistency}).
We emphasize that, to our knowledge, this is the first time that the limit consistency concept is used in reduced order modeling.
We also illustrated numerically the verifiability and limit consistency of the L-ROM and S-ROM for two under-resolved convection-dominated problems that display sharp gradients:
(i) the 1D Burgers equation with a small diffusion coefficient; and 
(ii) the 2D lid-driven cavity flow at Reynolds number $Re=15,000$.
The theoretical 
results for verifiability (Theorem~\ref{thm:verifiability}) and limit consistency (Theorem~\ref{thm:LimitConsistency}) 
were recovered for both the L-ROM and S-ROM in our numerical investigation. 

There are several research directions that could be followed to further develop both theoretically and  computationally the new L-ROM and the classical S-ROM.
First, numerical investigation of the new L-ROM in more challenging computational settings (e.g., the turbulent channel flow \cite{tsai2025time}) is required in order to assess the capabilities and limitations of the L-ROM.
For example, the new L-ROM could be compared with other LES-ROMs, e.g., data-driven EV LES-ROMs \cite{prakash2024projection,protas2015optimal}, and also ROM stabilizations \cite{Quaini2024Bridging,tsai2025time}.
One could also further develop the numerical analysis of the L-ROM, including both {\it a priori} bounds (that could be used, e.g., to find parameter scalings~\cite{reyes2024trrom}) and {\it a posteriori} bounds. 
Furthermore, one could adapt the current verifiability and limit consistency investigation to alternative LES-ROM frameworks based on different filters, e.g., by replacing the ROM projection, $P_r$, with a differential filter (whose radius could be used as a lengthscale in the construction of the L-ROM and S-ROM). This filter might additionally or alternatively incorporate a lengthscale that is more relevant to ROMs, such as the energy 
lengthscale formulation given in \cite{mou2023energy}.
Finally, the verifiability and limit consistency of other LES-ROMs should also be performed.

\bibliographystyle{plain}
\bibliography{Ref}

\begin{bibdiv}
\begin{biblist}

\bib{agdestein2025discretize}{article}{
      author={Agdestein, Syver~D{\o}ving},
      author={Sanderse, Benjamin},
       title={Discretize first, filter next: Learning divergence-consistent closure models for large-eddy simulation},
        date={2025},
     journal={J. Comput. Phys.},
      volume={522},
       pages={113577},
}

\bib{ahmed2018stabilized}{article}{
      author={Ahmed, M.},
      author={San, O.},
       title={Stabilized principal interval decomposition method for model reduction of nonlinear convective systems with moving shocks},
        date={2018},
     journal={Comp. Appl. Math.},
      volume={37},
      number={5},
       pages={6870\ndash 6902},
}

\bib{ahmed2021closures}{article}{
      author={Ahmed, Shady~E},
      author={Pawar, Suraj},
      author={San, Omer},
      author={Rasheed, Adil},
      author={Iliescu, Traian},
      author={Noack, Bernd~R},
       title={On closures for reduced order models—a spectrum of first-principle to machine-learned avenues},
        date={2021},
     journal={Phys. Fluids},
      volume={33},
      number={9},
}

\bib{ballarin2020certified}{article}{
      author={Ballarin, Francesco},
      author={Rebollo, Tom{\'a}s~Chac{\'o}n},
      author={Avila, Enrique~Delgado},
      author={M{\'a}rmol, Macarena~G{\'o}mez},
      author={Rozza, Gianluigi},
       title={Certified reduced basis {VMS}-{S}magorinsky model for natural convection flow in a cavity with variable height},
        date={2020},
     journal={Comput. Math. Appl.},
      volume={80},
      number={5},
       pages={973\ndash 989},
}

\bib{beirao2005regularity}{article}{
      author={Beir{\~a}o~da Veiga, Hugo},
       title={On the regularity of flows with {L}adyzhenskaya shear-dependent viscosity and slip or non-slip boundary conditions},
        date={2005},
     journal={Communications on Pure and Applied Mathematics: A Journal Issued by the Courant Institute of Mathematical Sciences},
      volume={58},
      number={4},
       pages={552\ndash 577},
}

\bib{bergmann2009enablers}{article}{
      author={Bergmann, M.},
      author={Bruneau, C.~H.},
      author={Iollo, A.},
       title={{Enablers for robust POD models}},
        date={2009},
     journal={J. Comput. Phys.},
      volume={228},
      number={2},
       pages={516\ndash 538},
}

\bib{berkooz1993proper}{article}{
      author={Berkooz, G.},
      author={Holmes, P.},
      author={Lumley, J.L.},
       title={The proper orthogonal decomposition in the analysis of turbulent flows},
        date={1993},
     journal={Ann. Rev. Fluid Mech.},
      volume={25},
      number={1},
       pages={539\ndash 575},
}

\bib{berselli2006mathematics}{book}{
      author={Berselli, L.~C.},
      author={Iliescu, T.},
      author={Layton, W.~J.},
       title={Mathematics of large eddy simulation of turbulent flows},
      series={Scientific Computation},
   publisher={Springer-Verlag},
     address={Berlin},
        date={2006},
        ISBN={978-3-540-26316-6; 3-540-26316-0},
      review={\MR{MR2185509 (2006h:76071)}},
}

\bib{cao2022continuous}{article}{
      author={Cao, Yu},
      author={Giorgini, Andrea},
      author={Jolly, Michael},
      author={Pakzad, Ali},
       title={Continuous data assimilation for the 3{D} {L}adyzhenskaya model: analysis and computations},
        date={2022},
     journal={Nonlinear Analysis: Real World Applications},
      volume={68},
       pages={103659},
}

\bib{chaturantabut2012state}{article}{
      author={Chaturantabut, Saifon},
      author={Sorensen, Danny~C},
       title={A state space error estimate for {POD}-{DEIM} nonlinear model reduction},
        date={2012},
     journal={SIAM J. Numer. Anal.},
      volume={50},
      number={1},
       pages={46\ndash 63},
}

\bib{CSB03}{article}{
      author={Couplet, M.},
      author={Sagaut, P.},
      author={Basdevant, C.},
       title={Intermodal energy transfers in a proper orthogonal decomposition--{G}alerkin representation of a turbulent separated flow},
        date={2003},
        ISSN={0022-1120},
     journal={J. Fluid Mech.},
      volume={491},
       pages={275\ndash 284},
      review={\MR{MR2014314}},
}

\bib{du1990finite}{article}{
      author={Du, Qiang},
      author={Gunzburger, Max~D},
       title={Finite-element approximations of a {L}adyzhenskaya model for stationary incompressible viscous flow},
        date={1990},
     journal={SIAM J. Numer. Anal.},
      volume={27},
      number={1},
       pages={1\ndash 19},
}

\bib{du1991analysis}{article}{
      author={Du, Qiang},
      author={Gunzburger, Max~D},
       title={Analysis of a {L}adyzhenskaya model for incompressible viscous flow},
        date={1991},
     journal={J. Math. Anal. Appl.},
      volume={155},
      number={1},
       pages={21\ndash 45},
}

\bib{GPMC91}{article}{
      author={Germano, M.},
      author={Piomelli, U.},
      author={Moin, P.},
      author={Cabot, W.H.},
       title={A dynamic subgrid-scale eddy viscosity model},
        date={1991},
     journal={Phys. Fluids A},
      volume={3},
       pages={1760\ndash 1765},
}

\bib{heywood1990finite}{article}{
      author={Heywood, John~G},
      author={Rannacher, Rolf},
       title={Finite-element approximation of the nonstationary {N}avier--{S}tokes problem. {Part IV}: Error analysis for second-order time discretization},
        date={1990},
     journal={SIAM J. Numer. Anal.},
      volume={27},
      number={2},
       pages={353\ndash 384},
}

\bib{reyesgsm}{article}{
      author={Huang, Shen~C.},
      author={Johnson, Adam},
      author={Neda, Monika},
      author={Reyes, Jorge},
      author={Tehrani, Hossein},
       title={A generalization of the {S}magorinsky model},
        date={2024},
        ISSN={0096-3003},
     journal={Appl. Math. Comput.},
      volume={469},
       pages={128545},
}

\bib{HughesVMS1998}{article}{
      author={Hughes, Thomas~J.R.},
      author={Feijóo, Gonzalo~R.},
      author={Mazzei, Luca},
      author={Quincy, Jean-Baptiste},
       title={The variational multiscale method—a paradigm for computational mechanics},
        date={1998},
        ISSN={0045-7825},
     journal={Comput. Methods Appl. Mech. Engrg.},
      volume={166},
      number={1},
       pages={3\ndash 24},
         url={https://www.sciencedirect.com/science/article/pii/S0045782598000796},
        note={Advances in Stabilized Methods in Computational Mechanics},
}

\bib{Iliescu_al18}{article}{
      author={Iliescu, T.},
      author={Liu, H.},
      author={Xie, X.},
       title={{Regularized reduced order models for a stochastic Burgers equation}},
        date={2018},
     journal={International Journal of Numerical Analysis \& Modeling},
      volume={15},
       pages={594\ndash 607},
}

\bib{iliescu2014variational}{article}{
      author={Iliescu, T.},
      author={Wang, Z.},
       title={Variational multiscale proper orthogonal decomposition: {N}avier-{S}tokes equations},
        date={2014},
     journal={Num. Meth. P.D.E.s},
      volume={30},
      number={2},
       pages={641\ndash 663},
}

\bib{ingimarson2022full}{article}{
      author={Ingimarson, Sean},
      author={Rebholz, Leo~G},
      author={Iliescu, Traian},
       title={Full and reduced order model consistency of the nonlinearity discretization in incompressible flows},
        date={2022},
     journal={Comput. Meth. Appl. Mech. Engrg.},
      volume={401},
       pages={115620},
}

\bib{iollo2000stability}{article}{
      author={Iollo, A.},
      author={Lanteri, S.},
      author={D{\'e}sid{\'e}ri, J.~A.},
       title={{Stability properties of POD--Galerkin approximations for the compressible Navier--Stokes equations}},
        date={2000},
     journal={Theoret. Comput. Fluid Dyn.},
      volume={13},
      number={6},
       pages={377\ndash 396},
}

\bib{john2003large}{book}{
      author={John, Volker},
       title={Large eddy simulation of turbulent incompressible flows: analytical and numerical results for a class of {LES} models},
   publisher={Springer Science \& Business Media},
        date={2003},
      volume={34},
}

\bib{kaneko2020towards}{article}{
      author={Kaneko, K.},
      author={Tsai, P.-H.},
      author={Fischer, P.},
       title={Towards model order reduction for fluid-thermal analysis},
        date={2020},
     journal={Nucl. Eng. Des.},
      volume={370},
       pages={110866},
}

\bib{Kaya2002verifiability}{article}{
      author={Kaya, M.},
      author={Layton, W.~J.},
       title={{On ``verifiability" of models of the motion of large eddies in turbulent flows}},
        date={2002},
     journal={Differential and Integral Equations},
      volume={15},
      number={11},
       pages={1395 \ndash  1407},
         url={https://doi.org/10.57262/die/1356060729},
}

\bib{koc2019commutation}{article}{
      author={Koc, B.},
      author={Mohebujjaman, M.},
      author={Mou, C.},
      author={Iliescu, T.},
       title={Commutation error in reduced order modeling of fluid flows},
        date={2019},
     journal={Adv. Comput. Math.},
      volume={45},
      number={5-6},
       pages={2587\ndash 2621},
}

\bib{koc2022verifiability}{article}{
      author={Koc, Birgul},
      author={Mou, Changhong},
      author={Liu, Honghu},
      author={Wang, Zhu},
      author={Rozza, Gianluigi},
      author={Iliescu, Traian},
       title={Verifiability of the data-driven variational multiscale reduced order model},
        date={2022},
     journal={J. Sci. Comput.},
      volume={93},
      number={2},
       pages={54},
}

\bib{KV99}{article}{
      author={Kunisch, K.},
      author={Volkwein, S.},
       title={Control of the {B}urgers equation by a reduced-order approach using proper orthogonal decomposition},
        date={1999},
        ISSN={0022-3239},
     journal={J. Optim. Theory Appl.},
      volume={102},
      number={2},
       pages={345\ndash 371},
      review={\MR{MR1706822 (2000e:49042)}},
}

\bib{KV01}{article}{
      author={Kunisch, K.},
      author={Volkwein, S.},
       title={Galerkin proper orthogonal decomposition methods for parabolic problems},
        date={2001},
        ISSN={0029-599X},
     journal={Numer. Math.},
      volume={90},
      number={1},
       pages={117\ndash 148},
      review={\MR{MR1868765 (2003g:65118)}},
}

\bib{ladyzhenskaya1985boundary}{book}{
      author={Ladyzhenskaya, O.A.},
  translator={Lohwater, J.},
       title={The boundary value problems of mathematical physics},
      series={Applied Mathematical Sciences},
   publisher={Springer},
     address={New York},
        date={1985},
      volume={49},
}

\bib{ladyzhenskaya1969mathematical}{article}{
      author={Ladyzhenskaya, Olga~Aleksandrovna},
       title={The mathematical theory of viscous incompressible flow},
        date={1969},
     journal={Gordon \& Breach},
}

\bib{layton2008introduction}{book}{
      author={Layton, William},
       title={Introduction to the numerical analysis of incompressible viscous flows},
   publisher={SIAM},
        date={2008},
      volume={6},
}

\bib{lilly1967representation}{inproceedings}{
      author={Lilly, Douglas~K},
       title={The representation of small-scale turbulence in numerical simulation experiments},
        date={1967},
   booktitle={Proc. ibm sci. comput. symp. on environmental science},
       pages={195\ndash 210},
}

\bib{lions1969quelques}{article}{
      author={Lions, Jacques~Louis},
       title={Quelques m{\'e}thodes de r{\'e}solution des problemes aux limites non lin{\'e}aires},
        date={1969},
}

\bib{manti2025symbolic}{article}{
      author={Manti, Simone},
      author={Tsai, Ping-Hsuan},
      author={Lucantonio, Alessandro},
      author={Iliescu, Traian},
       title={Symbolic regression of data-driven reduced order model closures for under-resolved, convection-dominated flows},
        date={2025},
     journal={arXiv preprint arXiv:2502.04703},
}

\bib{minty1962monotone}{article}{
      author={Minty, George~J},
       title={Monotone (nonlinear) operators in {H}ilbert space},
        date={1962},
     journal={Duke Math. J.},
      volume={29},
      number={3},
       pages={341\ndash 346},
}

\bib{mou2023energy}{article}{
      author={Mou, Changhong},
      author={Merzari, Elia},
      author={San, Omer},
      author={Iliescu, Traian},
       title={An energy-based lengthscale for reduced order models of turbulent flows},
        date={2023},
        ISSN={0029-5493},
     journal={Nucl. Eng. Des.},
      volume={412},
       pages={112454},
         url={https://www.sciencedirect.com/science/article/pii/S0029549323003035},
}

\bib{noack2002low}{techreport}{
      author={Noack, B.~R.},
      author={Papas, P.},
      author={Monkewitz, P.~A.},
       title={Low-dimensional {G}alerkin model of a laminar shear-layer},
 institution={{\'E}cole Polytechnique F{\'e}d{\'e}rale de Lausanne},
        date={2002},
      number={2002-01},
}

\bib{parish2024residual}{article}{
      author={Parish, E.~J.},
      author={Yano, M.},
      author={Tezaur, I.},
      author={Iliescu, T.},
       title={Residual-based stabilized reduced-order models of the transient convection-diffusion-reaction equation obtained through discrete and continuous projection},
        date={2024},
     journal={Arch. Comput. Methods Eng.},
       pages={1\ndash 45},
}

\bib{prakash2024projection}{article}{
      author={Prakash, A.},
      author={Zhang, Y.~J.},
       title={Projection-based reduced order modeling and data-driven artificial viscosity closures for incompressible fluid flows},
        date={2024},
     journal={Comput. Meth. Appl. Mech. Engrg.},
      volume={425},
       pages={116930},
}

\bib{protas2015optimal}{article}{
      author={Protas, B.},
      author={Noack, B.~R.},
      author={{\"O}sth, J.},
       title={Optimal nonlinear eddy viscosity in {G}alerkin models of turbulent flows},
        date={2015},
     journal={J. Fluid Mech.},
      volume={766},
       pages={337\ndash 367},
}

\bib{Quaini2024Bridging}{article}{
      author={Quaini, Annalisa},
      author={San, Omer},
      author={Veneziani, Alessandro},
      author={Iliescu, Traian},
       title={Bridging large eddy simulation and reduced-order modeling of convection-dominated flows through spatial filtering: Review and perspectives},
        date={2024},
        ISSN={2311-5521},
     journal={Fluids},
      volume={9},
      number={8},
         url={https://www.mdpi.com/2311-5521/9/8/178},
}

\bib{rebollo2014mathematical}{book}{
      author={Rebollo, T.~Chac{\'o}n},
      author={Lewandowski, R.},
       title={Mathematical and numerical foundations of turbulence models and applications},
   publisher={Springer},
        date={2014},
}

\bib{rebollo2023certified}{article}{
      author={Rebollo, Tom{\'a}s~Chac{\'o}n},
      author={{\'A}vila, Enrique~Delgado},
      author={M{\'a}rmol, Macarena~G{\'o}mez},
       title={On a certified {VMS-S}magorinsky reduced basis model with {LPS} pressure stabilisation},
        date={2023},
     journal={Appl. Numer. Math.},
      volume={185},
       pages={365\ndash 385},
}

\bib{rebollo2017certified}{article}{
      author={Rebollo, Tom{\'a}s~Chac{\'o}n},
      author={Avila, Enrique~Delgado},
      author={M{\'a}rmol, Macarena~G{\'o}mez},
      author={Ballarin, Francesco},
      author={Rozza, Gianluigi},
       title={On a certified {S}magorinsky reduced basis turbulence model},
        date={2017},
     journal={SIAM J. Numer. Anal.},
      volume={55},
      number={6},
       pages={3047\ndash 3067},
}

\bib{reyes2024trrom}{article}{
      author={Reyes, Jorge},
      author={Tsai, Ping-Hsuan},
      author={Novo, Julia},
      author={Iliescu, Traian},
       title={A priori error bounds and parameter scalings for the time relaxation reduced order model},
        date={2024},
     journal={arXiv preprint arXiv:2411.08986},
}

\bib{reyes2020projection}{article}{
      author={Reyes, Ricardo},
      author={Codina, Ramon},
       title={Projection-based reduced order models for flow problems: A variational multiscale approach},
        date={2020},
     journal={Comput. Meth. Appl. Mech. Engrg.},
      volume={363},
       pages={112844},
}

\bib{sagaut2006large}{book}{
      author={Sagaut, Pierre},
       title={Large eddy simulation for incompressible flows: an introduction},
   publisher={Springer Science \& Business Media},
        date={2006},
}

\bib{sanfilippo2023approximate}{article}{
      author={Sanfilippo, Anna},
      author={Moore, Ian},
      author={Ballarin, Francesco},
      author={Iliescu, Traian},
       title={Approximate deconvolution {L}eray reduced order model for convection-dominated flows},
        date={2023},
        ISSN={0168-874X},
     journal={Finite Elem. Anal. Des.},
      volume={226},
       pages={104021},
         url={https://www.sciencedirect.com/science/article/pii/S0168874X23001142},
}

\bib{smagorinsky1963general}{article}{
      author={Smagorinsky, Joseph},
       title={General circulation experiments with the primitive equations: I. the basic experiment},
        date={1963},
     journal={Monthly weather review},
      volume={91},
      number={3},
       pages={99\ndash 164},
}

\bib{stabile2019reduced}{article}{
      author={Stabile, Giovanni},
      author={Ballarin, Francesco},
      author={Zuccarino, Giacomo},
      author={Rozza, Gianluigi},
       title={A reduced order variational multiscale approach for turbulent flows},
        date={2019},
     journal={Adv. Comput. Math.},
      volume={45},
       pages={2349\ndash 2368},
}

\bib{strazzullo2022consistency}{article}{
      author={Strazzullo, Maria},
      author={Girfoglio, Michele},
      author={Ballarin, Francesco},
      author={Iliescu, Traian},
      author={Rozza, Gianluigi},
       title={Consistency of the full and reduced order models for evolve-filter-relax regularization of convection-dominated, marginally-resolved flows},
        date={2022},
     journal={Int. J. Num. Meth. Eng.},
      volume={123},
      number={14},
       pages={3148\ndash 3178},
}

\bib{temam2001navier}{book}{
      author={Temam, R.},
       title={{N}avier-{S}tokes equations: {T}heory and numerical analysis},
   publisher={American Mathematical Society},
        date={2001},
      volume={2},
}

\bib{tsai2025time}{article}{
      author={Tsai, P.~H.},
      author={Fischer, P.},
      author={Iliescu, T.},
       title={A time-relaxation reduced order model for the turbulent channel flow},
        date={2025},
     journal={J. Comput. Phys.},
      volume={521},
       pages={113563},
}

\bib{tsai2022parametric}{article}{
      author={Tsai, P.H.},
      author={Fischer, P.},
       title={Parametric model-order-reduction development for unsteady convection},
        date={2022},
     journal={Front. Phys.},
      volume={10},
       pages={903169},
}

\bib{tsai2023accelerating}{article}{
      author={Tsai, Ping-Hsuan},
      author={Fischer, Paul},
      author={Solomonik, Edgar},
       title={Accelerating the {G}alerkin reduced-order model with the tensor decomposition for turbulent flows},
        date={2023},
     journal={arXiv preprint arXiv:2311.03694},
}

\bib{ullmann2010pod}{inproceedings}{
      author={Ullmann, Sebastian},
      author={Lang, Jens},
       title={A {POD}-{G}alerkin reduced model with updated coefficients for {S}magorinsky {LES}},
        date={2010},
   booktitle={V {E}uropean conference on computational fluid dynamics, {ECCOMAS} {CFD}},
       pages={2010},
}

\bib{volkwein2013proper}{article}{
      author={Volkwein, S.},
       title={Proper orthogonal decomposition: Theory and reduced-order modelling},
        date={2013},
     journal={Lecture Notes, University of Konstanz},
        note={\url{http://www.math.uni-konstanz.de/numerik/personen/volkwein/teaching/POD-Book.pdf}},
}

\bib{wang2012proper}{article}{
      author={Wang, Zhu},
      author={Akhtar, Imran},
      author={Borggaard, Jeff},
      author={Iliescu, Traian},
       title={Proper orthogonal decomposition closure models for turbulent flows: a numerical comparison},
        date={2012},
     journal={Comput. Meth. Appl. Mech. Engrg.},
      volume={237},
       pages={10\ndash 26},
}

\bib{xie2017approximate}{article}{
      author={Xie, X.},
      author={Wells, D.},
      author={Wang, Z.},
      author={Iliescu, T.},
       title={Approximate deconvolution reduced order modeling},
        date={2017},
     journal={Comput. Methods Appl. Mech. Engrg.},
      volume={313},
       pages={512\ndash 534},
}

\bib{xie2018numerical}{article}{
      author={Xie, X.},
      author={Wells, D.},
      author={Wang, Z.},
      author={Iliescu, T.},
       title={Numerical analysis of the {L}eray reduced order model},
        date={2018},
     journal={J. Comput. Appl. Math.},
      volume={328},
       pages={12\ndash 29},
}

\end{biblist}
\end{bibdiv}

\end{document}